\documentclass[acmtog, nonacm,natbib=false]{acmart}
\AtBeginDocument{%
  \providecommand\BibTeX{{%
    \normalfont B\kern-0.5em{\scshape i\kern-0.25em b}\kern-0.8em\TeX}}}

\setcopyright{acmlicensed}
\copyrightyear{2024}
\acmYear{2024}
\acmDOI{XXXXXXX.XXXXXXX}

\acmConference[Conference acronym 'XX]{Make sure to enter the correct
  conference title from your rights confirmation email}{June 03--05,
  2018}{Woodstock, NY}
\acmBooktitle{Tallinn '24: Thirty-Ninth Annual ACM/IEEE Symposium on Logic in Computer Science (LICS),
 June 03--05, 2018, Tallinn, Estonia} 
\acmISBN{978-1-4503-XXXX-X/18/06}

\usepackage{quiver}
\usepackage{lettrine}
\usepackage{stmaryrd}
\usepackage{bbm}
\usepackage{fourier-orns}

\usepackage{mathpartir}

\newtheoremstyle{note}
  {2pt} % Space above
  {2pt} % Space below
  {} % Body font
  {1em} % Indent amount
  {\itshape} % Theorem head font
  {.} % Punctuation after theorem head
  {0.5em} % Space after theorem head
  {} % Theorem head spec (can be left empty, meaning `normal')

\theoremstyle{remark}

\theoremstyle{note}
\newtheorem*{note}{Note}

\newtheoremstyle{dfn}
  {2pt} % Space above
  {2pt} % Space below
  {} % Body font
  {1em} % Indent amount
  {\bfseries} % Theorem head font
  {} % Punctuation after theorem head
  {0.5em} % Space after theorem head
  {} % Theorem head spec (can be left empty, meaning `normal')

\theoremstyle{dfn}
\newtheorem{dfn}{Definition}[section]

\theoremstyle{dfn}
\newtheorem{example}{Example}[section]

\newcommand{\llp}[1]{\llparenthesis #1 \rrparenthesis}
\newcommand{\llb}[1]{\llbracket #1 \rrbracket}

\renewcommand\footnotetextcopyrightpermission[1]{}

\usepackage[style=alphabetic]{biblatex}
\bibliography{main}

\usepackage{hyperref}
\usepackage{xurl}
\hypersetup{breaklinks=true}

\begin{document}

%%
%% The "title" command has an optional parameter,
%% allowing the author to define a "short title" to be used in page headers.
\title{\aldine ~ Foundations of Substructural Dependent Type Theory}
\author{C.B. Aberlé}
\email{caberle@andrew.cmu.edu}
\affiliation{\institution{Carnegie Mellon University} \country{USA}}

%%
%% The "author" command and its associated commands are used to define
%% the authors and their affiliations.
%% Of note is the shared affiliation of the first two authors, and the
%% "authornote" and "authornotemark" commands
%% used to denote shared contribution to the research.

%%
%% By default, the full list of authors will be used in the page
%% headers. Often, this list is too long, and will overlap
%% other information printed in the page headers. This command allows
%% the author to define a more concise list
%% of authors' names for this purpose.

%%
%% The abstract is a short summary of the work to be presented in the
%% article.

\begin{abstract}
  \textbf{Abstract:} This paper presents preliminary work on a general system for integrating dependent types into substructural type systems such as linear logic and linear type theory. Prior work on this front has generally managed to deliver type systems possessing either syntax or semantics inclusive of certain practical applications, but has struggled to combine these all in one and the same system. Toward resolving this difficulty, I propose a novel categorical interpretation of substructural dependent types, analogous to the use of monoidal categories as models of linear and ordered logic, that encompasses a wide class of mathematical and computational examples. On this basis, I develop a general framework for substructural dependent type theories, and proceed to prove some essential metatheoretic properties thereof. As an application of this framework, I show how it can be used to construct a type theory that satisfactorily addresses the problem of effectively representing cut admissibility for linear sequent calculus in a logical framework.
\end{abstract}

%%
%% The code below is generated by the tool at http://dl.acm.org/ccs.cfm.
%% Please copy and paste the code instead of the example below.
%%

%%
%% Keywords. The author(s) should pick words that accurately describe
%% the work being presented. Separate the keywords with commas.

%% A "teaser" image appears between the author and affiliation
%% information and the body of the document, and typically spans the
%% page.

%%
%% This command processes the author and affiliation and title
%% information and builds the first part of the formatted document.
\maketitle

\section{The Past, Present \& Future of Substructural Dependent Type Theory}

\lettrine{D}{ependent type theory} in the mould of Martin-Löf's \emph{intuitionistic type theory} \cite{martin-lof} promises to internalize mathematical reasoning into systems for constructing proof and program alike. Likewise, substructural type systems akin to Girard's \emph{linear logic} \cite{girard} seek to reflect the fundamental insight that \emph{truth is ephemeral}, and to thereby capture notions of \emph{state, resources, etc.} in the proofs/programs they afford. Yet these two typing disciplines, both alike in dignity, have proven resistant to a satisfactory union.

Many authors seeking to combine substructural and dependent typing have followed the solution posed by Cervesato \& Pfenning \cite{cervesato-pfenning}, whereby contexts of would-be substructural dependent type theories are to be cleft in twain – on one side an \emph{intuitionistic} context of variables upon which others may depend, and on the other a context of \emph{substructural} variables amongst which there inheres no dependency. This approach is not without theoretical merit nor practical utility, but by its very nature, it falls short of fully capturing substructural reasoning about proofs/programs, since substructurality must end where dependency begins. As Jason Reed observes \cite{reed}, this has the practical consequence that one cannot give an effective representation of cut-admissibility for linear sequent calculus in a linear logical framework constructed in this way. Such a representation must reflect substructural constraints on terms into the types in which those terms appear – precisely what is disallowed by forfeiting dependency to the intuitionistic layer.

Other solutions to the problem of combining substructural and dependent types have been proposed, e.g. forms of \emph{quantitative type theory} \cite{mcbride} \cite{atkey}. However, the majority of these share with Cervesato \& Pfenning's the confinement of dependent type formation to an essentially intuitionistic layer of the theory, which thereby runs afoul of the same example highlighted above by Reed.

What is distinctly lacking from the above attempts at integrating substructural and dependent types is a solidly agreed-upon semantic basis in mathematics upon which to build the type-theoretic apparatus. To be sure, many of the above-mentioned dependent type theories have been given denotational semantics, some quite elegant, often as variations on the categorical semantics of dependent type theory, with some additional structure to make sense of the substructural component (e.g. \cite{krishnaswami-pradic-benton}). Yet, as illustrated by such troublesome cases as Reed's, none of these semantics capture the full generality of substructural dependent types in the same way that, e.g., the usual categorical semantics for linear and ordered logic in monoidal categories do. Upon reflection, it should be clear why this would be the case -- the semantics of linear logic were arrived at not by \emph{adding} structure to the semantics of intuitionistic type theory, but rather by carefully \emph{removing} such structure.

I thus propose to develop a truly general substructural dependent type theory by generalizing the categorical semantics of dependent type theory to a novel class of structures, which I call \emph{left-fibred double categories} (LFDCs), as a unifying framework for both dependent and substructural type theories. I establish metatheoretic desiderata, such as decidability, admissibility of substitution, etc., for the resulting type theory, and finally, to illustrate its general applicability, I turn it upon the example of representing cut-admissibility for linear sequent calculus, and show how the type theory succeeds at this task where others have previously failed.

\section{Sketch of a Semantics}

I begin with a reflection upon and generalization of the categorical semantics of dependent types. For present purposes, I keep the exposition somewhat informal, with fully rigorous definitions and proofs to be conducted in future work.

\subsection{Revisting the Semantics of Dependent Type Theory}

Toward defining a class of structures capable of modelling substructural dependent types, I first recall the ordinary categorical semantics of dependent types. Various definitions exist in the literature of categorical structures that model dependent types, e.g. categories with families, display map categories, natural models, etc. For present purposes, however, it will suffice for us to consider only the most general class of such models -- \emph{comprehension categories}.

\begin{dfn}
\label{compcat} A \emph{comprehension category}, viewed as a model of dependent type theory (c.f. Jacobs \cite{jacobs}), comprises:

\begin{enumerate}
\item A category $\mathfrak{C}$ of \emph{contexts}, with objects $\Gamma, \Delta$ representing contexts and morphisms $f : \Gamma \xrightarrow[]{} \Delta \in \mathfrak{C}$ viewed as \emph{substitutions}.

\item For each context $\Gamma \in \mathfrak{C}$, a category $\mathfrak{T}(\Gamma)$ of \emph{types}, with objects $S,T$ interpreted as types dependent upon $\Gamma$ and morphisms $t : S \xrightarrow[]{} T \in \mathfrak{T}(\Gamma)$ as terms of type $T$ in context $\Gamma$, additionally parameterized by $S$.

\item For each $f : \Gamma \xrightarrow[]{} \Delta \in \mathfrak{C}$, a functor $f^* : \mathfrak{T}(\Delta) \to \mathfrak{T}(\Gamma)$ representing the \emph{application} of the substitution $f$, such that the assignment of functors $f^*$ preserves identities and composites of substitutions up to coherent natural isomorphism.

\item For each $\Gamma \in \mathfrak{C}$, a \emph{context extension} functor $\Gamma_\bullet : \mathfrak{T}(\Gamma) \to \mathfrak{C}$ mapping each type $S$ to the context $\Gamma_\bullet(S)$, to be thought of as $\Gamma$ \emph{extended} with a fresh variable of type $S$.

\item For each substitution $f : \Gamma \xrightarrow[]{} \Delta \in \mathfrak{C}$, a natural transformation $f_\bullet$ with components $f_\bullet(S) : \Gamma_\bullet(f^*(S)) \xrightarrow[]{} \Delta_\bullet(S) \in \mathfrak{C}$ for all $S \in \mathfrak{T}(\Delta)$, which -- roughly speaking -- applies $f$ to turn $\Gamma$ \emph{under} an extension of $\Gamma$ substituted over $f$. We moreover require that the assignment of natural transformations $f_\bullet$ preserves identities and composites modulo the natural isomorphisms described above in item 3.

\item For each context $\Gamma \in \mathfrak{C}$, a natural transformation $\pi_\Gamma$ with components $\pi_\Gamma^S : \Gamma_\bullet(S) \xrightarrow[]{} \Gamma \in \mathfrak{C}$ for each $S \in \mathfrak{T}(\Gamma)$, to be thought of as \emph{projecting out} $\Gamma$ from an extension of $\Gamma$, such that for each substitution $f : \Gamma \xrightarrow[]{} \Delta \in \mathfrak{C}$ and type $S \in \mathfrak{T}(\Delta)$, the following square is a pullback: \[\begin{tikzcd}
	{\Gamma_\bullet(f^*(S))} & {\Delta_\bullet(S)} \\
	\small\Gamma & {\Delta}
	\arrow["{f_\bullet(S)}", from=1-1, to=1-2]
	\arrow["{\pi_\Gamma}"', from=1-1, to=2-1]
	\arrow["{f}"', from=2-1, to=2-2]
	\arrow["{\pi_\Delta}", from=1-2, to=2-2]
	\arrow["\lrcorner"{anchor=center, pos=0.125}, draw=none, from=1-1, to=2-2]
\end{tikzcd}\]
\end{enumerate}
\end{dfn}

A reader familiar with other categorical models of type theory, e.g. categories with families, may puzzle over the fact that our presentation involves interpreting terms as morphisms $f : S \xrightarrow[]{} T \in \mathfrak{T}(\Gamma)$, parameterized by both the context $\Gamma$ and an additional type $S$. In fact, this extra parameterization is crucial for generalizing the semantics to substructural dependent types and effectively separates the "type-level" and "term-level" aspects of contexts.

In a comprehension category, the projection maps $\pi_\Gamma^S$ induce "term-level weakening," whereby we may discard the extended part of a context. Substituting along these projection maps yields corresponding \emph{``type-level weakening''} functors $\omega_\Gamma^S : \mathfrak{T}(\Gamma) \to \mathfrak{T}(\Gamma_\bullet(S))$, which play a key role in defining additional type-theoretic structure on comprehension categories, specifically:

\begin{dfn} A comprehension category has \emph{Dependent Sums}:

\begin{enumerate}
\item if for each context $\Gamma \in \mathfrak{C}$ and type $S \in \mathfrak{T}(\Gamma)$, the weakening functor $\omega_\Gamma^S$ has a left adjoint $\Sigma_S : \mathfrak{T}(\Gamma_\bullet(S)) \to \mathfrak{T}(S)$, i.e. there is a natural bijection of homsets: $$\text{Hom}_{\mathfrak{T}(\Gamma_\bullet(S))}(T,\omega_\Gamma^S(R)) \cong \text{Hom}_{\mathfrak{T}(\Gamma)}(\Sigma_S(T), R)$$ \item and if the functors $\Sigma_S$ satisfy the \emph{Beck-Chevalley} condition, which essentially states that for any ${f : \Gamma \to \Delta \in \mathfrak{C}}$ along with $S \in \mathfrak{T}(\Delta)$ and $T \in \mathfrak{T}(\Delta_\bullet(S))$, there is a canonical isomorphism $f^*(\Sigma_S(T)) \simeq \Sigma_{f^*(S)} (f_\bullet(S)^*(T))$.
\end{enumerate}

\noindent\label{strongsums}A comprehension category with dependent sums as above moreover has \emph{strong sums} if the canonical substitution $$(\Gamma_\bullet(S))_\bullet(T) \xrightarrow[]{} \Gamma_\bullet(\Sigma_S(T)) \in \mathfrak{C}$$ is an isomorphism for all $\Gamma \in \mathfrak{C}$ with $S \in \mathfrak{T}(\Gamma)$ and $T \in \mathfrak{T}(\Gamma_\bullet(S))$.
\end{dfn}
From the first item in the above definition, one may deduce the usual rules for weak sum types -- and if the comprehension category has strong sums, one may likewise deduce the usual rules for strong sum types. The second item -- the Beck-Chevalley condition -- is necessary for substitution to behave as expected when substituting into dependent sums, i.e. that substituting into a dependent sum yields a dependent sum. More generally, such compositionality of substitution ought to be required of any additional type-theoretic structure we impose upon a comprehension category.

\emph{Dependent products} (i.e. dependent function types) in a comprehension category are similarly defined as \emph{right}-adjoints to type-level weakening that satisfy an analogous Beck-Chevalley condition. There is also a notion of \emph{unit types} in a comprehension category:

\begin{dfn}
A comprehension category has \emph{unit types} if, for each context $\Gamma \in \mathfrak{C}$, the category $\mathfrak{T}(\Gamma)$ has a terminal object $\top_\Gamma$, and the substitution functors preserve terminal objects. These unit types are \emph{strong} if the projection $\pi_\Gamma : \Gamma_\bullet(\top_\Gamma) \xrightarrow[]{} \Gamma \in \mathfrak{C}$ is an isomorphism for all contexts $\Gamma \in \mathfrak{C}$.
\end{dfn}

In all of the above cases, we establish type-theoretic structure in comprehension categories by way of certain \emph{universal properties}. This is evident even in the definition of a comprehension category itself, where context extension gains a universal property through projection maps forming pullback squares with substitutions. This aligns with the \emph{intuitionistic} nature of ordinary dependent type theory. By way of analogy, context extension in the intuitionistic theory of simple types is given the universal property of a product. One then obtains models for linear logic \& linear type theory -- i.e. (closed, symmetric) monoidal categories -- by replacing the universal \emph{property} of context extension with the \emph{structure} of a monoidal product that is unital and associative up to coherent isomorphism.

To obtain models of substructural dependent type theory, we may try replacing the universal property of context extension in comprehension categories with appropriately unital and associative \emph{structure}. However, a challenge arises in stating unitality and associativity for contexts, since context extension builds up contexts one element at a time. A primitive semantic notion of \emph{partial context} or \emph{telescope} \cite{debruijn} is seemingly needed to break up contexts for stating associativity. In ordinary models of dependent type theory, the role of such partial contexts is effectively played by \emph{strong sums} -- the condition for strong sums given in \textbf{Def. \ref{strongsums}} can be seen as implicitly capturing a kind of associativity between dependent sums and context extension. This suggests that the proper semantic setting for substructural dependent type theory ought to be a generalization of \emph{comprehension categories with strong sums} (and strong unit types) wherein the universal properties of both context extension and the dependent pair/unit types are discarded in favor of structure witnessing their mutual associativity and unitality. Such categorical structures, which I call left-fibred double categories (LFDCs), shall be my object of study for the remainder of this section.

\subsection{Left-Fibred Double Categories}

\begin{dfn}
An LFDC consists of the same data as items 1-5 of \textbf{Def. \ref{compcat}}, along with the following:

\begin{enumerate}
\item For each context $\Gamma \in \mathfrak{C}$ and type $S \in \mathfrak{T}(\Gamma)$, a \emph{dependent pair type functor} ${\oplus_S : \mathfrak{T}(\Gamma_\bullet(S)) \to \mathfrak{T}(\Gamma)}$.
\item For each substitution $f : \Gamma \to \Delta \in \mathfrak{C}$ and each term $g : f^*(S) \xrightarrow[]{} S' \in \mathfrak{T}(\Gamma)$, a natural transformation $\langle g,- \rangle$ with components $\oplus_{f^*(S)}(\Gamma_\bullet(g)^*(T)) \to \oplus_{S'}(T)$.
\item For each context $\Gamma \in \mathfrak{C}$, an \emph{unit type} object $\mathbbm{1}_\Gamma \in \mathfrak{T}(\Gamma)$.
\item Natural isomorphisms $\eta, \rho, \ell, \alpha, \beta, \gamma, \delta$ with components: \begin{enumerate}
    \item $\eta_\Gamma : \Gamma_\bullet(\mathbbm{1}_\Gamma) \simeq \Gamma$ for all $\Gamma \in \mathfrak{C}$
    \item $\rho_{\Gamma, S} : \oplus_S(\mathbbm{1}_{\Gamma_\bullet(S)}) \simeq S$ for all $\Gamma \in \mathfrak{C}$ and $S \in \mathfrak{T}(\Gamma)$
    \item $\ell_{\Gamma,S} : \oplus_{\mathbbm{1}_\Gamma}(\eta_\Gamma^*(S)) \simeq S$ for all $\Gamma \in \mathfrak{C}$ and $S \in \mathfrak{T}(\Gamma)$
    \item $\alpha_{\Gamma, S, T} : \Gamma_\bullet(\oplus_S(T)) \simeq (\Gamma_\bullet(S))_\bullet(T)$ for all $\Gamma \in \mathfrak{C}$ together with $S \in \mathfrak{T}(\Gamma)$ and $T \in \mathfrak{T}(\Gamma_\bullet(S))$
    \item $\beta_{\Gamma, S, T, R} : \oplus_S(\oplus_T(R)) \simeq \oplus_{\oplus_S(T)}(\alpha_{\Gamma,S,T}^*(R))$ for all\\ $\Gamma \in \mathfrak{C}, ~ S \in \mathfrak{T}(\Gamma)$ with $T \in \mathfrak{T}(\Gamma_\bullet(S))$ and $R \in \mathfrak{T}((\Gamma_\bullet (S))_\bullet T)$
    \item $\gamma_{\Gamma,\Delta,f,S,T} : f^*(\oplus_S(T)) \simeq \oplus_{f^*(S)}(f_\bullet(S)^*(T))$ for all $\Gamma, \Delta \in \mathfrak{C}$ with ${f : \Gamma \xrightarrow[]{} \Delta} \in\mathfrak{C}$ and $S \in \mathfrak{T}(\Delta)$ and $T \in \mathfrak{T}(\Delta_\bullet(S))$ (Beck-Chevalley for pair types)
    \item $\delta_{\Gamma,\Delta,f} : f^*(\mathbbm{1}_\Delta) \simeq \mathbbm{1}_\Gamma$ for all $\Gamma, \Delta \in \mathfrak{C}$ with ${f : \Gamma \xrightarrow[]{} \Delta} \in \mathfrak{C}$ (Beck-Chevalley for unit types)
\end{enumerate} subject to certain coherence laws.
\end{enumerate}
\end{dfn}
The above definition, up to and including item 4(e), is equivalently described as a \emph{pseudomonad} in a certain tricategory (specifically the tricategory whose objects are categories and whose 1-cells are spans of categories with one leg a fibration). This alternative definition makes clear that the above structure is a kind of \emph{double category}, with the property that its domain projection functor is a fibration, whence the name \emph{left-fibred double category}.

I am aware of only one place in the category-theoretic literature where structures such as these have been studied before, which is in recent work by David Jaz Myers \& Matteo Cappucci \cite{cappucci-myers}, as part of the two authors' work on Categorical Systems Theory, wherein they refer to such structures as ``dependent actegories.'' Suffice it to say that, to my knowledge, the use of LFDCs as models of type theory is novel to this paper.

\begin{dfn}
One additional piece of type-theoretic structure that has so-far been missing from the above exposition of LFDCs is some notion of \emph{empty context} from which to build other contexts. To this end, we have the following: a \emph{unit context} in an LFDC is a context $\epsilon \in \mathfrak{C}$ such that the context extension functor ${\epsilon_\bullet}$ is an equivalence of categories between $\mathfrak{T}(\epsilon)$ and $\mathfrak{C}$. We think of types $T \in \mathfrak{T}(\epsilon)$ as \emph{closed types}, such that contexts and closed types may be treated interchangeably. Hence we shall hereafter restrict our attention mainly to LFDCs equipped with a choice of unit context.
\end{dfn}

\begin{example} \label{syncats}
Any comprehension category with strong sums and strong unit types is inherently an LFDC, making all models of ordinary dependent type theory naturally LFDCs. Specifically, any intuitionistic dependent type theory with strong sum types and a unit type can be straightforwardly transformed into a comprehension category (the \emph{syntactic category} of the type theory) with strong sums/unit types (hence an LFDC) by quotienting terms up to $\alpha$-equivalence of $\beta\eta$-normal forms, and defining substitution, dependent pair types, etc., as their syntactic counterparts.
\end{example}

\begin{example} \label{moncat-lfdc}
The LFDC concept also encompasses categorical models of linear/ordered logic, specifically monoidal categories. In this context, a monoidal category $\mathcal{M}$ with monoidal product $\otimes : \mathcal{M} \times \mathcal{M} \to \mathcal{M}$ and unit $I \in \mathcal{M}$ is interpreted both as the category of contexts and as the category of types for each context $\Gamma \in \mathcal{M}$. The substitution functors are then simply the identity on $\mathcal{M}$, and $\Gamma_\bullet(S)$ and $\oplus_S(T)$ are interpreted as $\Gamma \otimes S$ and $S \otimes T$, respectively, with unit types interpreted as $I$.
\end{example}

Observe that in an LFDC based on a monoidal category $\mathcal{M}$, a term $t : S \xrightarrow[]{} T \in \mathfrak{T}(\Gamma)$ only trivially depends upon $\Gamma$, since substitution into $\Gamma$ is given by the identity, while substituting into the term-level context $S$ involves potentially non-trivial composition of morphisms in $\mathcal{M}$. LFDCs arising from monoidal categories thus exhibit non-trivial substructural properties, but trivial type dependency, whereas LFDCs arising from comprehension categories exhibit non-trivial type dependency, but trivial substructural properties. As an example of an LFDC with both non-trivial type-dependency and substructural behavior, we have the following:

\begin{example}
The category of linearly ordered sets (with either strictly or weakly order-preserving maps as morphisms) is naturally regarded as an LFDC, by interpreting both context extension and dependent pair types as the \emph{lexicographic sum} of linear orders. A type in context $(\Gamma, <_\Gamma)$ is defined as a family of linear orders $\{(S_x,<x)\}_{x \in \Gamma}$ indexed by the set $\Gamma$. Given such a family of linear orders, we define the context extension $(\Gamma,<_\Gamma)_\bullet(\{(S_x, <_x)\}_{x \in \Gamma})$ as the \emph{lexicographic sum} of the family, i.e. the set $\{(x,s) \mid x \in \Gamma, s \in S_x\}$ ordered such that $(x,s) < (x',s')$ iff either $x <_\Gamma x'$ or $x = x'$ and $s <_x s'$. We define the dependent pair type similarly as a \emph{family} of lexicographic sums, indexed by the underlying set of its context.
\end{example}

This example straightforwardly exhibits both non-trivial type dependency and substructutral properties -- on the one hand, \emph{any} family of linear orders indexed by the underlying set of another linear order counts as a dependent type in this setting, while on the other hand, the lexicographic product, which is an instance of the lexicographic sum over a constant family of linear orders, is notably non-symmetric, and so is distinct from the ordinary product of intuitionistic type theory.

Although the above argument is intuitively clear, we do not yet possess the appropriate structure on LFDCs to phrase such arguments in the internal language of LFDCs themselves -- e.g. we defined the lexicographic product of linear orders as the lexicographic sum over a constant family, i.e. a family not dependent upon its context; As it stands, the basic definition of LFDC given above poses no criterion for when a type may be considered \emph{independent} from (part of) its context. If we examine the above-given examples, we find that they all admit such a notion of \emph{independence} or \emph{type-level weakening}. It thus seems appropriate to augment the definition of an LFDC with such a notion of \emph{type-level weakening}, which is in fact critical to defining further type-theoretic structure on LFDCs.

\subsection{Type-Level Weakening}

\begin{dfn}
An LFDC has \emph{type-level weakening} if it is additionally equipped with:

\begin{enumerate}
\item A functor $\omega_\Gamma^S : \mathfrak{T}(\Gamma) \to \mathfrak{T}(\Gamma_\bullet(S))$ for each context $\Gamma \in \mathfrak{C}$ and type $S \in \mathfrak{T}(\Gamma)$
\item A functor $\Omega_\Gamma^{S,T} : \mathfrak{T}(\Gamma_\bullet T) \to \mathfrak{T}((\Gamma_\bullet S)_\bullet \omega_\Gamma^S(T))$ for each context $\Gamma \in \mathfrak{C}$ and types $S,T \in \mathfrak{T}(\Gamma)$
\item Subject to certain natural isomorphisms and coherence laws thereupon that ensure the compatibility of the above-given functors with each other and with the structure of the ambient LFDC. In particular, there are natural isomorphisms $\zeta$ and $\theta$ with components: \begin{itemize} \item $\zeta_{f,g,T} : \Gamma_\bullet(g)^*(\omega_\Gamma^{S'}(T)) \simeq \omega^{f^*(S)}_\Gamma(T)$ \item $\theta_{f,g,T,R} : \Gamma_\bullet(g)_\bullet(\omega_{\Gamma}^{S'}(T))^*(\Omega_{\Gamma}^{S',T}(R)) \simeq \zeta_{f,g,T}^*(\Omega_{\Gamma}^{f^*(S),T}(R))$ \end{itemize} for all $\Gamma, \Delta \in \mathfrak{C}$ with $f : \Gamma \to \Delta \in \mathfrak{C}$ and $S',T,R \in \mathfrak{T}(\Gamma)$ and $S \in \mathfrak{T}(\Delta)$ along with $g : f^*(S) \xrightarrow[]{} S' \in \mathfrak{T}(\Gamma)$. These isomorphisms effectively witness that a weakened type \emph{does not} depend upon the type added to its context by weakening. We also require natural isomorphisms $\xi, \chi$ with components \begin{itemize} \item $\xi_{\Gamma,S} : \omega_\Gamma^S(\mathbbm{1}_\Gamma) \simeq \mathbbm{1}_{\Gamma_\bullet(S)}$ \item $\chi_{\Gamma,S,T,R} : \omega_\Gamma^S(\oplus_T(R)) \simeq \oplus_{\omega_\Gamma^S(T)}(\Omega_{\Gamma}^{S,T}(R))$ \end{itemize} for all $\Gamma \in \mathfrak{C}$ with $S,T \in \mathfrak{T}(\Gamma)$ and $R \in \mathfrak{T}(\Gamma_\bullet(T))$. These isomorphisms ensure that weakening behaves as we would expect with regard to the type formers available in the ambient LFDC, i.e. the weakening of a unit type is itself (up to isomorphism) a unit type, and likewise for dependent pair types. When positing additional type formers for an LFDC with type-level weakening, we shall therefore generally require these to come equipped with analogous natural isomorphisms to ensure their compatibility with the type-level weakening structure, along with the usual Beck-Chevalley isomorphisms.
\end{enumerate}
\end{dfn}

Analogous to the definition of the lexicographic product as the lexicographic sum of a constant family of linear orders, we then have the following in any LFDC with type-level weakening:

\begin{dfn} \label{weak-mon}
In an LFDC with type-level weakening, each category of types $\mathfrak{T}(\Gamma)$ for each context $\Gamma$ is naturally equipped with the structure of a monoidal category whose unit object is $\mathbbm{1}_\Gamma$, where the monoidal product $S \otimes_\Gamma T$ of $S, T \in \mathfrak{T}(\Gamma)$ is defined as the dependent pair type $\oplus_S(\omega_\Gamma^S(T))$.
\end{dfn}

Given type-level weakening functors and the associated monoidal products on an LFDC, it becomes straightforward to define function types as suitable right adjoints to these. Since LFDCs distinguish type-level and term-level contexts, we should expect there to correspondingly be notions of both \emph{term-level} and \emph{type-level} function types in an LFDC. We thus have the following definitions:

\begin{dfn}
An LFDC with type-level weakening has \emph{term-level function types} if for each context $\Gamma \in \mathfrak{C}$ with $S \in \mathfrak{T}(\Gamma)$, the functors $- \otimes_\Gamma S$ and $S \otimes_\Gamma -$ have right adjoints $\sslash_S$ and $\bbslash_S$, respectively, with the associated natural bijections of homsets: $$
\small \sslash_{\Gamma,S,T,R} : \text{Hom}_{\mathfrak{T}(\Gamma)}(R \otimes S, T) \cong \text{Hom}_{\mathfrak{T}(\Gamma)}(R, \sslash_S(T))$$ $$ \bbslash_{\Gamma,S,T,R} :
\text{Hom}_{\mathfrak{T}(\Gamma)}(S \otimes R, T) \cong \text{Hom}_{\mathfrak{T}(\Gamma)}(R, \bbslash_S(T))
$$ We additionally require the functors $\sslash_S$ and $\bbslash_S$ to satisfy a Beck-Chevalley condition and compatibility with type-level weakening in the form of canonical isomorphisms $$f^*(\sslash_S(T)) \simeq \sslash_{f^*(S)}(f^*(T)) \quad f^*(\bbslash_S(T)) \simeq \bbslash_{f^*(S)}(f^*(T))$$ $$\omega_\Gamma^R(\sslash_S(T)) \simeq \sslash_{\omega_\Gamma^R(S)}(\omega_\Gamma^R(T)) \quad \omega_\Gamma^R(\bbslash_S(T)) \simeq \bbslash_{\omega_\Gamma^R(S)}(\omega_\Gamma^R(T))$$
etc.
\end{dfn}

\begin{dfn}
an LFDC with type-level weakening has \emph{type-level function types} if for each context $\Gamma \in \mathfrak{C}$ with $S \in \mathfrak{T}(\Gamma)$, the weakening functor $\omega_\Gamma^S : \mathfrak{T}(\Gamma) \to \mathfrak{T}(\Gamma_\bullet(S))$ has a right adjoint $\forall_S$, with corresponding natural bijection of homsets $$
\Lambda_{\Gamma,S,T,R} : \text{Hom}_{\mathfrak{T}(\Gamma_\bullet(S))}(\omega_\Gamma^S(R),T) \cong \text{Hom}_{\mathfrak{T}(\Gamma)}(R,\forall_S(T))
$$ As in the above definition, we also require the functors $\forall_S$ to be compatible with substitution and type-level weakening in that there are canonical isomorphisms $$f^*(\forall_S(T)) \simeq \forall_{f^*(S)}(f_\bullet(S)^*(T)) \quad \omega_\Gamma^R(\forall_S(T)) \simeq \forall_{\omega_\Gamma^R(S)}(\Omega_\Gamma^{R,S}(T))$$
\end{dfn} \noindent and so on. In an LFDC with a unit context $\epsilon$ and both term-level and type-level function types, ``open'' terms $S \xrightarrow[]{} T \in \mathfrak{T}(\Gamma)$ are equivalent to \emph{closed} terms of type $\mathbbm{1}_\epsilon \to_{\mathfrak{T}(\epsilon)}\forall_{\epsilon_\bullet^{-1}(\Gamma)}(\sslash_S(T))$ (where $\epsilon_\bullet^{-1}$ is the inverse to context extension by $\epsilon$). Such LFDCs thus fully internalize their logic of term formation.

Moving on to \emph{structural properties} of LFDCs with type-level weakening, among the most important of these is \emph{exchange}, which allows permuting independent variables in contexts. Having defined the monoidal structure on types arising from type-level weakening, we are now in a position to state this property for LFDCs in general:

\begin{dfn} \label{exchange}
an LFDC with type-level weakening has \emph{exchange} if the mon\-oidal structure defined on each category of types $\mathfrak{T}(\Gamma)$ in \textbf{Def. \ref{weak-mon}} additionally carries the structure of a \emph{symmetric monoidal category}, i.e. a natural isomorphism $\sigma$ with components $$\sigma_{\Gamma,S,T} : S \varotimes_\Gamma T \simeq T \varotimes_\Gamma S$$ for all $S,T \in \mathfrak{T}(\Gamma)$, subject to certain coherence conditions.
\end{dfn}

Beyond exchange, there is also \emph{term-level weakening}, which allows discarding variables in the term-level context. This property may be defined for \emph{any} LFDC, with or without type-level weakening, and in fact equips the LFDC with a natural form of type-level weakening.

\begin{dfn} \label{term-weak}
An LFDC has \emph{term-level weakening} if for each context $\Gamma$, the type $\mathbbm{1}_\Gamma$ is a terminal object in the category $\mathfrak{T}(\Gamma)$. For any context $\Gamma$, write $\top_S : S \xrightarrow[]{} \mathbbm{1}_\Gamma \in \mathfrak{T}(\Gamma)$ for the unique morphism from any type $S \in \mathfrak{T}(\Gamma)$ to $\mathbbm{1}_\Gamma$.

Any LFDC with term-level weakening is naturally equpped with \emph{type-level weakening}, given by substitution along terminal morphisms. Specifically, given a context $\Gamma \in \mathfrak{C}$ and types $S,T \in \mathfrak{T}(\Gamma)$ and $R \in \mathfrak{T}(\Gamma_\bullet(T))$, we may define $$
\omega_\Gamma^S(T) = \Gamma_\bullet(\top_S)^*(\eta_\Gamma^*(T))$$ and $$\Omega_\Gamma^{S,T}(R) = (\alpha^{-1})^*(\Gamma_\bullet(\langle \top_S, - \rangle)^*(\ell_{\Gamma,T}^*(R)))$$

\noindent We may likewise define the following \emph{projection maps}: $$
\varoplus_S(T) \xrightarrow{\varoplus_S(\top_T)} \varoplus_S(\mathbbm{1}_{\Gamma_\bullet(S)}) \xrightarrow{\rho_{\Gamma,S}} S
$$ $$
\varoplus_S(\omega_\Gamma^S(T)) \xrightarrow{\langle \top_S, - \rangle} \varoplus_{\mathbbm{1}_\Gamma}(\eta^*_\Gamma(T)) \xrightarrow{\ell_{\Gamma,T}} T
$$ where $\omega_\Gamma^S(T)$ is as defined above. The latter projection induces a family of natural transformations $$\text{Hom}_{\mathfrak{T}(\Gamma_\bullet(S))}(T,\omega_\Gamma^S(R)) \to \text{Hom}_{\mathfrak{T}(\Gamma)}(\varoplus_S(T),R)$$ by mapping $f : T \to \omega_\Gamma^S(R) \in \mathfrak{T}(\Gamma_\bullet(S))$ to the composite $$
\varoplus_S(T) \xrightarrow{\varoplus_S(f)} \varoplus_S(\omega_\Gamma^S(R)) \xrightarrow{\langle \top_S, - \rangle} \varoplus_{\mathbbm{1}_\Gamma}(\eta_\Gamma^*(R)) \xrightarrow{\ell_{\Gamma,R}} R
$$ In this sense, type-level weakening in such an LFDC is ``halfway'' to forming an adjunction with the dependent pair type. This motivates the following definition of \emph{contraction} in an LFDC with type-level weakening as forming the other half of such an adjunction.
\end{dfn}

\begin{dfn} \label{contract}
An LFDC with type-level weakening has \emph{contraction} if there is a natural transformation of homsets: $$
\text{Hom}_{\mathfrak{T}(\Gamma)}(\varoplus_S(T),R) \to \text{Hom}_{\mathfrak{T}(\Gamma_\bullet(S))}(T,\omega_\Gamma^S(R))
$$ for all contexts $\Gamma \in \mathfrak{C}$ and types $S, R \in \mathfrak{T}(\Gamma)$ and $T \in \mathfrak{T}(\Gamma_\bullet(S))$, that is suitably compatible with the structure of the LFDC.

By a Yoneda-style argument, the above is equivalent to the existence of a natural family of morphisms $$\mathfrak{d}_{\Gamma,S} : \mathbbm{1}_{\Gamma_\bullet(S)} \xrightarrow[]{} \omega_\Gamma^S(S) \in \mathfrak{T}(\Gamma_\bullet(S))$$ for all contexts $\Gamma \in \mathfrak{C}$ with $S \in \mathfrak{T}(\Gamma)$, suitably compatible with the LFDC structure. In this sense, contraction allows for the term-level use of type-level variables. This makes sense if we think of type-level contexts in LFDCs as consisting of variables that are treated as \emph{already having been used elsewhere}.

It follows that in an LFDC with contraction, for each context $\Gamma \in \mathfrak{C}$, there is a natural transformation $\Delta$ with components $\Delta_S : S \xrightarrow[]{} S \otimes_\Gamma S \in \mathfrak{T}(\Gamma)$ for all types $S \in \mathfrak{T}(\Gamma)$, defined as follows: $$
S \xrightarrow{\rho_{\Gamma,S}^{-1}} \varoplus_S(\mathbbm{1}_{\Gamma_\bullet(S)}) \xrightarrow{\varoplus_S( \mathfrak{d}_{\Gamma_\bullet(S),S})} \varoplus_S(\omega_\Gamma^S(S)) = S \otimes_\Gamma S
$$ If we think of $\Delta_S$ as \emph{duplicating} its input, then this suffices to show that the above-defined notion of contraction for LFDCs gives rise to the usual notion of contraction as \emph{copying} of variables.
\end{dfn}

We may then define a \emph{Cartesian LFDC} as one which has both term-level weakening and contraction in a compatible way, specifically:

\begin{dfn} \label{cart}
An LFDC is \emph{Cartesian} if 1) it has term-level weakening, and 2) the natural transformations of homsets defined in \textbf{Def. \ref{term-weak}} are all invertible, i.e. the type-level weakening functors $\omega_\Gamma^S$ are right-adjoint to the dependent pair type functors $\varoplus_S$.

Moreover, any Cartesian LFDC is naturally equipped with \emph{exchange}, since for any context $\Gamma \in \mathfrak{C}$ and types $S,T \in \mathfrak{T}(\Gamma)$, we may define a morphism $\sigma : S \otimes_\Gamma T \to T \otimes_\Gamma S$ as follows: $$ S \otimes_\Gamma T \xrightarrow{\Delta_{S \otimes T}} \begin{array}{c} (S \otimes T)\\ \otimes\\ (S \otimes T) \end{array} \xrightarrow{(\top_S \otimes T) \otimes (S \otimes \top_T)} \begin{array}{c} (\mathbbm{1}_\Gamma \otimes T)\\ \otimes\\ (S \otimes \mathbbm{1}_\Gamma) \end{array} \simeq T \otimes S $$ where $\Delta$ is as defined in \textbf{Def. \ref{contract}}. One then checks that $\sigma$ is an isomorphism satisfying all the properties given in \textbf{Def. \ref{exchange}}.

More generally, the monoidal product $\otimes_\Gamma$ on each category of types $\mathfrak{T}(\Gamma)$ in a Cartesian LFDC satisfies the universal property of a \emph{product}. Even if the given LFDC is not Cartesian, it may have such products in its type-theoretic structure, without these necessarily coinciding with the monoidal product, as follows:
\end{dfn}

\begin{dfn}
An LFDC has \emph{product types} if: \begin{enumerate}
\item For each $\Gamma \in \mathfrak{C}$, there is a functor $\times_\Gamma : \mathfrak{T}(\Gamma) \times \mathfrak{T}(\Gamma) \to \mathfrak{T}(\Gamma)$ together with natural transformations $\pi_1, \pi_2$ with components $\pi_1^{S,T} : S \times_\Gamma T \to S$ and $\pi_2^{S,T} : S \times_\Gamma T \to T$, such that for all types $S,T,R \in \mathfrak{T}(\Gamma)$ with morphisms $f : R \xrightarrow[]{} S \in \mathfrak{T}(\Gamma)$ and $g : R \xrightarrow[]{} T \in \mathfrak{T}(\Gamma)$, there exists a unique morphism $(f,g) : R \xrightarrow[]{} S \times_\Gamma T \in \mathfrak{T}(\Gamma)$ which factors $f$ and $g$ through $\pi_1^{S,T}$ and $\pi_2^{S,T}$, respectively.
\item The substitution functors $f^*$ are all finite-product-preserving, and if the LFDC has type-level weakening, then similarly the type-level weakening functors preserve finite products.
\end{enumerate}

We think of the product $S \times_\Gamma T$ of two types $S,T \in \mathfrak{T}(\Gamma)$ as offering a \emph{choice} between an element of $S$ and an element of $T$, while the monoidal product $S \otimes_\Gamma T$ offers two such elements at once.
\end{dfn}

\subsection{Strictness}

Concluding the discussion on LFDCs, I now shift my attention to their role as models of type theory in particular. While the class of LFDCs has so far been defined as to include various mathematical examples, this definition is in a sense still too \emph{weak}. Specifically, we have so far required substitution and type-level weakening to satisfy their requisite identities only up to coherent isomorphism, when from the perspective of type theory, these should hold on the nose. For this purpose, I introduce the following definition:

\begin{dfn}
An LFDC is \emph{strict} if:
\begin{enumerate}
\item The assignment of substitution functors $f^*$ \emph{strictly} preserves identities and composites of substitutions.
\item The Beck-Chevalley isomorphisms $\gamma, \delta$ are identities, i.e. substitution \emph{strictly} preserves unit and dependent pair types.
\end{enumerate}

\noindent For LFDCs with type-level weakening, the natural isomorphisms associated with type-level weakening must also be strict. Similarly, substitution and type-level weakening in LFDCs with function types / products must strictly preserve the associated type formers.
\end{dfn}

\begin{example} \label{synmodels}
Most of the above-considered examples of LFDCs fail to be strict. However, one significant class of LFDCs manages to meet these requirements, which is the class of \emph{syntactic categories} of dependent type theory as in \textbf{Ex. \ref{syncats}}. Because (type-level) weakening in such models is given by inclusions of terms into extended contexts, and because substitution is essentially computed by applying the prescribed identities, it follows that these models are in fact strict.
\end{example}

\section{The Type Theory of LFDCs}

What is a type theory? Standardly, one might answer that a type theory is a definition of some syntax along with rules for forming judgments upon this syntax, and perhaps also for computing therewith. As an alternative, I wish to pose the view that a type theory is a system of translation between syntax and semantics. In the ordinary approach to type theory, the computational dynamics form an elementary semantics for the theory, defined e.g. in terms of equivalence classes of terms modulo $\alpha$-equivalence of normal forms, etc. Yet this may be generalized -- provided we have a procedure for translating well typed terms of the theory into some semantic domain, we need not restrict ourselves to semantics that are so-dependent upon the syntax of the theory itself. Instead, we may define a type theory parameterized by a class of models in which to interpret it, and then search among this class for a model that yields good computational behavior. This allows for a separation of concerns between the design of the syntax of the type theory and the assurance of its computational adequacy. \footnote{This perspective on type theories also generalizes the more modern approach of specifying the syntax and rules of a type theory within a logical framework using higher-order abstract syntax (HOAS). In that case, the semantic domain in which the type theory is interpreted is simply the ambient logical framework itself; although this offers significant advantages over the tradtional approach to the specification of type theories, it is not suitable for our present purposes of defining \emph{substructural} dependent type theory, since any type theory specified via HOAS in an intuitionistic logical framework will allow for intuitionistic use of variables in type families -- precisely what we wish to avoid in general. Hence I have been compelled to seek a suitable abstraction of this approach, which has led me to the above-described perspective on type theories, and the categorical semantics of substructural type theory sketched in the prior section.}

Such separation of concerns is particularly useful in the design of \emph{dependent} type theory. Typically in dependent type theory, checking whether a term belongs to a given type may require performing some amount of \emph{computation} within the type. Thus, type checking of terms depends upon \emph{semantic information} about the types. This suggests that, in a typing judgment such as $t : T$, the type $T$ ought to be an object drawn from the semantics of the theory, while the term $t$ (the subject of the judgment) is a syntactic representation of an inhabitant of $T$. The rules for checking such a judgment can then be designed in such a way as to implicitly \emph{compute} the semantic denotation of $t$. This is moreover in line with the discipline of \emph{bidirectional type-checking}, wherein the judgments of a type theory are conceived not merely as predicates, but as \emph{procedures} for computing certain information about terms/types. I hence adopt a bidirectional approach to the type theory I present herein, where every judgment computes the semantic denotation of its subject. 

Because we have taken care in the previous section to identify the type-theoretic structure of LFDCs (with type-level weakening, function types, etc.), it is largely straightforward to give a type-theoretic syntax for this structure, as I now illustrate.

\subsection{Expressions, Denotations, Signatures \& Judgments}

\begin{dfn} Assume a countably infinite stock of \emph{variables} -- which I shall write as $a,b,c,d,w,x,y,z$ -- and countably infinite supplies of \emph{type family symbols} and $\emph{constant symbols}$ -- written $P,Q$ and $p,q$, respectively. We then have the following grammar of syntactic \emph{expressions}:$$
\begin{array}{rclll}
\text{Contexts} ~ \Gamma,\Delta,\Theta,\Phi & ::= & \epsilon & \mid & \Gamma, x : S
\end{array} 
$$ $$
\begin{array}{rclllllll}
\text{Types} ~ S,T,R,U,V & ::= & \mathbbm{1}
& \mid & \bigoplus x : S . T\\
& \mid & S \sslash  T
& \mid & T \bbslash S\\
& \mid & \forall x : S . T
& \mid & S \times T
& \mid & Ps
\end{array}
$$ $$
\begin{array}{rclllllll}
\text{Introduction Forms} ~ s,t,r,u,v & ::= & \langle \rangle
& \mid & \langle s,t \rangle\\
& \mid & \sslash x . t
& \mid & \bbslash x . t\\
& \mid & \Lambda x . t
& \mid & (s,t)
& \mid & \underline{e}
\end{array}$$ $$\begin{array}{l}
\text{Elimination Forms} ~ e,f\\
\begin{array}{cl} 
::= &
\mathsf{let}\left[_{R}\right] ~ \langle \rangle = e ~ \mathsf{in} ~ r\\ 
\mid & 
\mathsf{let}\left[a._R^{U}\right] ~ \langle \rangle = e ~ \mathsf{with} ~ w = u ~ \mathsf{in} ~ r\\
\mid & \mathsf{let}\left[a,b._{R}^{U}\right] ~ \langle \rangle = e ~ \mathsf{with} ~ w = u ~ \mathsf{in} ~ r\\
\mid & \mathsf{let}\left[a,b,c._{R}^{U, ~ V}\right] ~ \langle \rangle = e ~ \mathsf{with} ~ w = u ~ \mathsf{and} ~ z = v ~ \mathsf{in} ~ r\\
\mid & \mathsf{let}\left[a._{R}^{U} \right] ~ \langle x,y \rangle = e ~ \mathsf{with} ~ w = u ~ \mathsf{in} ~ r\\
\mid & \mathsf{let}\left[a,b,c._{R}^{U, ~ V} \right] ~ \langle x,y \rangle = e ~ \mathsf{with} ~ w = u ~ \mathsf{and} ~ z = v ~ \mathsf{in} ~ r\\
\mid & \begin{array}{lllllllll}
f \triangleleft s 
& \mid & s \triangleright f 
& \mid & f \cdot s
& \mid & \pi_1(e)
& \mid & \pi_2(e)
\end{array}\\
\mid & \begin{array}{llllll}
t : T & \mid & x & \mid & p
\end{array}
\end{array}
\end{array}
$$ \end{dfn}

\begin{dfn}
Fix a strict LFDC with type-level weakening, function types, and products. For each syntactic sort of expressions defined above, we then define a corresponding class of \emph{denotations}:

\begin{itemize}
\item \textbf{Contexts}: the denotation of a context $\Gamma$ is a list $\llparenthesis \Gamma \rrparenthesis$ of type annotations $x : \llbracket S \rrbracket$, which is well-formed if each such $\llbracket S \rrbracket$ is a type in the context corresponding to the list of annotations preceding $x$. We then have the following definition of the semantic context $\llbracket \llparenthesis \Gamma \rrparenthesis \rrbracket$ represented by $\llparenthesis \Gamma \rrparenthesis$ $$\llbracket \epsilon \rrbracket = \epsilon \qquad \llbracket \llparenthesis \Gamma \rrparenthesis, x : \llbracket S \rrbracket \rrbracket = \llbracket \llparenthesis \Gamma \rrparenthesis \rrbracket_\bullet(\llbracket S \rrbracket)$$ I write $\llparenthesis \Gamma \rrparenthesis, \llparenthesis \Delta \rrparenthesis$ (or $\llparenthesis \Gamma \rrparenthesis \fatsemi \llparenthesis \Delta \rrparenthesis$ when $\llparenthesis \Gamma \rrparenthesis$ is understood to be the type-level context, and $\llp{\Delta}$ the term-level context), to mean that 1) $\llparenthesis \Gamma \rrparenthesis$ is well-formed, and 2) the concatenation of $\llparenthesis \Delta \rrparenthesis$ onto the right-hand side of $\llparenthesis \Gamma \rrparenthesis$ is also well-formed, and say that $\llparenthesis \Delta \rrparenthesis$ is a \emph{telescope} for $\llparenthesis \Gamma \rrparenthesis$. We then have the following recursive definition of the semantic \emph{type} $\llbracket \llparenthesis \Delta \rrparenthesis \rrbracket_{\llparenthesis \Gamma \rrparenthesis}$ represented by a telescope $\llparenthesis \Delta \rrparenthesis$ for a semantic context $\llparenthesis \Gamma \rrparenthesis$: $$\qquad \llbracket \epsilon \rrbracket_{\llparenthesis \Gamma \rrparenthesis} = \mathbbm{1}_{\llbracket \llparenthesis \Gamma \rrparenthesis \rrbracket} \quad \llbracket x : \llbracket S \rrbracket, \llparenthesis \Delta \rrparenthesis \rrbracket_{ \llparenthesis \Gamma \rrparenthesis } = \varoplus_{\llbracket S \rrbracket}(\llbracket \llparenthesis \Delta \rrparenthesis \rrbracket_{ \llparenthesis \Gamma \rrparenthesis, x : \llbracket S \rrbracket})$$
\item \textbf{Types}: given a semantic context $\llparenthesis \Gamma \rrparenthesis$, the denotation of a type $S$ wrt $\llparenthesis \Gamma \rrparenthesis$ is defined as an object $\llbracket S \rrbracket_{\llparenthesis \Gamma \rrparenthesis} \in \mathfrak{T}(\llbracket \llparenthesis \Gamma \rrparenthesis \rrbracket)$.
\item \textbf{Introduction \& Elimination Forms:} given a semantic context $\llparenthesis \Gamma \rrparenthesis$ and a telescope $\llparenthesis \Gamma \rrparenthesis \fatsemi \llparenthesis \Delta \rrparenthesis$ together with a semantic type $\llbracket S \rrbracket_{ \llparenthesis \Gamma \rrparenthesis}$, the denotation of an introduction/elimination form $\digamma$ wrt $\llparenthesis \Gamma \rrparenthesis \fatsemi \llparenthesis \Delta \rrparenthesis$ and $\llbracket S \rrbracket_{\llbracket \llparenthesis \Gamma \rrparenthesis \rrbracket}$ is defined as a morphism $\llbracket \digamma \rrbracket^{\llparenthesis \Gamma \rrparenthesis \fatsemi \llparenthesis \Delta \rrparenthesis}_{\llbracket S \rrbracket} : \llbracket \llparenthesis \Delta \rrparenthesis \rrbracket_{ \llparenthesis \Gamma \rrparenthesis } \xrightarrow[]{} \llbracket S \rrbracket_{ \llparenthesis \Gamma \rrparenthesis } \in \mathfrak{T}(\llbracket \llparenthesis \Gamma \rrparenthesis \rrbracket)$.
\end{itemize}
\end{dfn}

\begin{dfn}
In order to capture structure of the LFDC in which the type theory is interpreted, we allow for the type theory to additionally be parameterized by a \emph{signature} of \emph{constant terms} and \emph{atomic type families}. Such a signature $\Sigma$ is defined to be a \emph{set} of \emph{bindings}, where each binding has one of the following forms: \begin{itemize} \item $P(\llbracket S \rrbracket) \mapsto \llbracket P \rrbracket$, where $\llbracket S \rrbracket \in \mathfrak{T}(\epsilon)$ is a closed semantic type, and $\llbracket P \rrbracket \in \mathfrak{T}(\epsilon_\bullet(\llbracket S \rrbracket))$ is a type family dependent upon $\llbracket S \rrbracket$.
\item $p : \llbracket S \rrbracket \mapsto \llbracket p \rrbracket$, where $\llbracket S \rrbracket \in \mathfrak{T}(\epsilon)$ is a closed semantic type, and $\llbracket p \rrbracket : \mathbbm{1}_\epsilon \xrightarrow[]{} \llbracket S \rrbracket \in \mathfrak{T}(\epsilon)$ is a closed term of type $\llbracket S \rrbracket$.
\end{itemize}
\end{dfn}

\begin{dfn} \label{judgmental}
We conceive of judgments as procedures taking \emph{inputs} (here written in \textcolor{blue}{blue}) and yielding \emph{outputs} (here written in \textcolor{red}{red}), with one input of each judgment designated the \emph{subject} (written in \textcolor{violet}{violet}). I follow McBride \cite{mcbride} in writing inputs to a judgment to the left of the subject, and outputs to the right, using $\ni$ and $\in$ for the judgments in which types are checked and inferred, respectively. The four main judgments of the type theory are then as follows: $$
\textcolor{violet}{\Gamma} ~ \mathsf{Ctx} \gg \textcolor{red}{\llparenthesis \Gamma \rrparenthesis} \qquad \textcolor{blue}{\llparenthesis \Gamma \rrparenthesis} \vdash \textcolor{violet}{T} ~ \mathsf{Type} \gg \textcolor{red}{\llbracket T \rrbracket_{\llparenthesis \Gamma \rrparenthesis}}
$$ $$
\textcolor{blue}{\llparenthesis \Gamma \rrparenthesis \fatsemi \llparenthesis \Delta \rrparenthesis} \vdash \textcolor{blue}{\llbracket T \rrbracket_{ \llparenthesis \Gamma \rrparenthesis}} \ni \textcolor{violet}{t} \gg \textcolor{red}{\llbracket t \rrbracket^{\llparenthesis \Gamma \rrparenthesis \fatsemi \llparenthesis \Delta \rrparenthesis}_{\llbracket T \rrbracket}}$$ $$ \textcolor{blue}{\llparenthesis \Gamma \rrparenthesis \fatsemi \llparenthesis \Delta \rrparenthesis} \vdash \textcolor{violet}{e} \in \textcolor{red}{\llbracket R \rrbracket_{\llparenthesis \Gamma \rrparenthesis}} \gg \textcolor{red}{\llbracket e \rrbracket^{\llparenthesis \Gamma \rrparenthesis \fatsemi \llparenthesis \Delta \rrparenthesis}_{\llbracket R \rrbracket}}
$$ Observe that within these judgments all inputs/outputs other than the subject are \emph{semantic objects}. Hence it is only ever the subject of a judgment which requires checking, and all dependence of type checking upon semantic information about types/contexts/etc. is suitably encoded in the structure of these judgments.
\end{dfn}

\begin{note}
From this point on, I will generally omit the sub/superscripts from $\llbracket S \rrbracket_{\llparenthesis \Gamma \rrparenthesis}$ and $\llbracket s \rrbracket^{\llparenthesis \Gamma \rrparenthesis \fatsemi \llparenthesis \Delta \rrparenthesis}_{\llbracket S \rrbracket}$, as these are readily inferred from their occurrence in the above judgments. Similarly, I will omit sub/super-\\scripts from natural transformations where these may inferred.
\end{note}

The reader familiar with traditional expositions of type theory may notice a conspicuous lack of a judgment for \emph{equality} among those given above. In fact, there is no need for such a judgment, because we have defined the above judgments to compute the semantic denotations of their subjects -- we thus may simply consider types $S,T$ judgmentally equal whenever $\llb{S} = \llb{T}$, introduction forms $s,t$ judgmentally equal whenever $\llb{s} = \llb{t}$, etc.

\begin{dfn}
Before proceeding with the rules of the type theory, I first note some useful constructions on semantic contexts / types / terms (see Appendix A for full definitions):
\begin{itemize}
\item \emph{Functoriality of telescopes}: given a telescope $\llparenthesis \Gamma \rrparenthesis, \llparenthesis \Delta \rrparenthesis$, there is an associated functor $\varoplus_{\llparenthesis \Delta \rrparenthesis} : \mathfrak{T}(\llbracket \llparenthesis \Gamma \rrparenthesis, \llparenthesis \Delta \rrparenthesis \rrbracket) \to \mathfrak{T}(\llbracket \llparenthesis \Gamma \rrparenthesis \rrbracket)$ given by iterative formation of dependent pair types.
\item \emph{Context substitution/weakening}: given a telescope $\llp{\Delta},\llp{\Theta}$ and substitution $f : \llb{\llp{\Gamma}} \to \llb{\llp{\Delta}}$, there is a telescope $\llp{\Gamma}, f^*(\llp{\Theta})$ given by substituting into all types in $\llp{\Theta}$. Similarly, for a telescope $\llp{\Gamma}, \llp{\Delta}$ and type $\llb{S} \in \mathfrak{T}(\llb{\llp{\Gamma}})$, there is a telescope $\llp{\Gamma}, x : \llb{S}, \omega_{\llp{\Gamma}}^{\llb{S}}(\llp{\Delta})$ given by weakening all types in $\llp{\Delta}$.
\item \emph{Reassociation}: for each telescope $\llparenthesis \Gamma \rrparenthesis, \llparenthesis \Delta \rrparenthesis$ there is an isomorphism ${\bullet}\mathsf{asc}_{\llparenthesis \Gamma \rrparenthesis}^{\llparenthesis \Delta \rrparenthesis} : \llbracket \llparenthesis \Gamma \rrparenthesis, \llparenthesis \Delta \rrparenthesis \rrbracket \simeq  \llbracket \llparenthesis \Gamma \rrparenthesis \rrbracket_\bullet(\llbracket \llparenthesis \Delta \rrparenthesis \rrbracket_{\llparenthesis \Gamma \rrparenthesis})$. Similarly, for each telescope $\llparenthesis \Gamma \rrparenthesis, \llparenthesis \Delta \rrparenthesis$ and type $\llb{S} \in \mathfrak{T}(\llb{\llp{\Gamma}, \llp{\Delta}})$ there is an isomorphism $$\qquad \quad {\oplus}\mathsf{asc}_{\llp{\Gamma}}^{\llp{\Delta} \mid \llb{S}} : \varoplus_{\llp{\Delta}}(\llb{S}) \simeq \varoplus_{\llb{\llp{\Delta}}}((({\bullet}\mathsf{asc}_{\llp{\Gamma}}^{\llp{\Delta}})^{-1})^*(\llb{S}))$$
\item \emph{Term substitution}: given a term $\llbracket s \rrbracket : \llbracket \llparenthesis \Delta \rrparenthesis \rrbracket \xrightarrow[]{} \llbracket S \rrbracket \in \mathfrak{T}(\llbracket \llparenthesis \Gamma \rrparenthesis \rrbracket)$, there is a substitution ${\downarrow}(\llbracket s \rrbracket) : \llbracket \llparenthesis \Gamma \rrparenthesis, \llparenthesis \Delta \rrparenthesis \rrbracket \to \llbracket \llparenthesis \Gamma \rrparenthesis \rrbracket_\bullet(\llbracket S \rrbracket)$, and given a term $\llb{e} : \llb{\llp{\Delta}} \xrightarrow[]{} \llb{S} \in \mathfrak{T}(\llb{\llp{\Gamma}})$ and a telescope $\llp{\Gamma}, x : \llb{S}, \llp{\Theta}$, there is a substitution $$\qquad {\Downarrow}(\llb{e}) : \llb{\llp{\Gamma}, \llp{\Delta}, {\downarrow}(\llb{e})^*(\llp{\Theta})} \xrightarrow[]{} \llb{\llp{\Gamma}, x : \llb{S}, \llp{\Theta}}$$
\item \emph{Lifting}: for each telescope $\llparenthesis \Gamma \rrparenthesis, \llparenthesis \Delta \rrparenthesis$ and type $\llbracket S \rrbracket \in \mathfrak{T}(\llbracket \llparenthesis \Gamma \rrparenthesis)$, there are \emph{lifting} functors $\uparrow_{\llparenthesis \Gamma \rrparenthesis}^{\llparenthesis \Delta \rrparenthesis} : \mathfrak{T}(\llbracket \llparenthesis \Gamma \rrparenthesis \rrbracket) \to \mathfrak{T}(\llbracket \llparenthesis \Gamma \rrparenthesis, \llparenthesis \Delta \rrparenthesis \rrbracket)$ and $\Uparrow_{\llparenthesis \Gamma \rrparenthesis}^{\llparenthesis \Delta \rrparenthesis \mid \llbracket S \rrbracket} : \mathfrak{T}(\llbracket \llparenthesis \Gamma \rrparenthesis \rrbracket_\bullet(\llbracket S \rrbracket)) \to \mathfrak{T}(\llbracket \llparenthesis \Gamma \rrparenthesis, \llparenthesis \Delta \rrparenthesis \rrbracket_\bullet(\uparrow_{\llparenthesis \Gamma \rrparenthesis}^{\llparenthesis \Delta \rrparenthesis}(\llbracket S \rrbracket)))$ that lift types into extended contexts by iterative weakening, and a \emph{telescopic} weakening functor $$\qquad \mathsf{W}_{\llp{\Gamma}}^{\llb{S} \mid \llp{\Delta}} : \mathfrak{T}(\llb{\llp{\Gamma}, \llp{\Delta}}) \to \mathfrak{T}(\llb{\llp{\Gamma}, x : \llb{S}, \omega_{\llp{\Gamma}}^{\llb{S}}(\llp{\Delta})})$$
\item \emph{Projection}: in an LFDC with term-level weakening, for a telescope $\llp{\Gamma}, \llp{\Delta}$ and a type $\llb{S} \in \mathfrak{T}(\llb{\llp{\Gamma}})$, there is a \emph{projection} $$\mathsf{proj}_{\llp{\Gamma}}^{\llb{\Delta} \mid \llb{S}} : \llb{\llp{\Delta}, x : {\uparrow}_{\llp{\Gamma}}^{\llp{\Delta}}(\llb{S})} \xrightarrow[]{} \llb{S} \in \mathfrak{T}(\llb{\llp{\Gamma}})$$
\end{itemize}
\end{dfn}

\subsection{Inference Rules}

In the bidirectional treatment of judgments-as-programs, a rule defines a procedure for evaluating a judgment on a subset of its input domain. For a rule of the form $$
\inferrule{J_1\\ \dots\\ J_n}{J}
$$ where $J$ and $J_1, \dots, J_n$ are judgments, the algorithmic interpretation is bottom-to-top: to evaluate the \emph{conclusion} $J$, evaluate the \emph{premises} $J_1, \dots, J_n$ in order. McBride \cite{mcbride} offers the following heuristic – a rule is a \emph{server} for its conclusion and a \emph{client} for its premises, accepting inputs to its conclusion, supplying inputs to its premises, receiving outputs of the premises, and forming the output of its conclusion accordingly. Because of McBride's discipline of writing all inputs to the left and all outputs to the right of the subject of each judgment, information thus flows \emph{clockwise} around a rule, starting and ending at the 6 o'clock position.

Judgments are then evaluated by nondeterministically evaluating rules whose input patterns match the input to the judgment. If no rule applies, the judgment signals failure, which is propagated to any rule invoking it as a subroutine. A successful evaluation is given by a derivation tree built from the rules -- hence induction on derivations can be applied as usual in the metatheory. To minimize nondeterminism, we shall seek to ensure that at most one rule applies to any syntactic form of the subject of its conclusion, so that judgment evaluation is \emph{syntax-directed}.

\subsubsection{Context Formation} Because we have defined the denotations of contexts in close connection to their syntactic representations, the rules for context formation are straightforward: $$
\small \inferrule*[right=Emp]{~}{\epsilon ~ \mathsf{Ctx} \gg \epsilon} \quad \inferrule*[right=Ext]{\Gamma ~ \mathsf{Ctx} \gg \llparenthesis \Gamma \rrparenthesis\\ \llparenthesis \Gamma \rrparenthesis \vdash S ~ \mathsf{Type} \gg \llbracket S \rrbracket\\ x \notin \Gamma}{\Gamma, x : S ~ \mathsf{Ctx} \gg \llparenthesis \Gamma \rrparenthesis, x : \llbracket S \rrbracket}
$$ 

\subsubsection{Annotation \& Embedding} We have the following rules for switching between type checking and inference. The rule for switching from inference to checking requires us to \emph{annotate} the term to be checked with the type against which to check it: $$
\small \inferrule*[Right=Annotate]{\llparenthesis \Gamma \rrparenthesis \vdash T ~ \mathsf{Type} \gg \llbracket T \rrbracket\\ \llparenthesis \Gamma \rrparenthesis \fatsemi \llparenthesis \Delta \rrparenthesis \vdash \llbracket T \rrbracket \ni t \gg \llbracket t \rrbracket}{\llparenthesis \Gamma \rrparenthesis \fatsemi \llparenthesis \Delta \rrparenthesis \vdash t : T \in \llbracket T \rrbracket \gg \llbracket t \rrbracket}
$$ while the rule for switching from checking to inference instead requires that the type we are able to synthesize for a given term is equal to the one already provided for checking: $$
\small \inferrule*[right=Embed]{\llparenthesis \Gamma \rrparenthesis \fatsemi \llparenthesis \Delta \rrparenthesis \vdash e \in \llbracket R \rrbracket \gg \llbracket e \rrbracket\\ \llbracket R \rrbracket = \llbracket T \rrbracket}{\llparenthesis \Gamma \rrparenthesis \fatsemi \llparenthesis \Delta \rrparenthesis \vdash \llbracket T \rrbracket \ni \underline{e} \gg \llbracket e \rrbracket}
$$

\subsubsection{Atoms \& Constants} Fix a signature $\Sigma$. We then have the following rule for instantiating atomic type families contained in $\Sigma$: $$
\small \inferrule*[right=Atom]{P(\llbracket S \rrbracket) \mapsto \llbracket P \rrbracket \in \Sigma \\\llparenthesis \Gamma \rrparenthesis \fatsemi \llparenthesis \Delta \rrparenthesis \vdash {\uparrow_\epsilon^{\llparenthesis \Gamma \rrparenthesis}}(\llbracket S \rrbracket) \ni s \gg \llbracket s \rrbracket}{\llparenthesis \Gamma \rrparenthesis, \llparenthesis \Delta \rrparenthesis, \llparenthesis \Theta \rrparenthesis \vdash Ps ~ \mathsf{Type} \gg \uparrow_{\llparenthesis \Gamma \rrparenthesis, \llparenthesis \Delta \rrparenthesis}^{\llparenthesis \Theta \rrparenthesis}({\downarrow} (\llbracket s \rrbracket)^* (\Uparrow_{\epsilon}^{\llparenthesis \Gamma \rrparenthesis \mid \llbracket S \rrbracket}(\llbracket P \rrbracket)))}
$$ Due to type-level weakening, we may make use of some \emph{part} $\llp{\Delta}$ of the context in forming the parameter $s$ for a type family $P$, discarding $\llp{\Theta}$ while keeping $\llp{\Gamma}$ at the type level. Note that once the part of the context $\llp{\Delta}$ used in constructing $s$ has been selected, the term-level typing rules enforce that $\llp{\Delta}$ must be used in accordance with the substructural constraints of the type theory. Similarly, we have the following rule for instantiating constants from $\Sigma$: $$
\small \inferrule*[Right=Const]{p : \llbracket S \rrbracket \mapsto \llbracket p \rrbracket \in \Sigma}{\llparenthesis \Gamma \rrparenthesis \fatsemi \epsilon \vdash p \in {\uparrow_\epsilon^{\llparenthesis \Gamma \rrparenthesis}}(\llbracket S \rrbracket) \gg \uparrow_\epsilon^{\llparenthesis \Gamma \rrparenthesis}(\llbracket p \rrbracket)}
$$

\subsubsection{Identity, Weakening \& Contraction} If the ambient LFDC has neither term-level weakening nor contraction, then we have only the following \emph{strict identity} rule: $$
\small \inferrule*[right=StrId]{~}{\llparenthesis \Gamma \rrparenthesis \fatsemi x : \llbracket S \rrbracket \vdash x \in \llbracket S \rrbracket \gg \text{id}_{\llbracket S \rrbracket}}
$$ If, on the other hand, the ambient LFDC has term-level weakening, then we may instead make use of the following \emph{weak identity} rule: $$
\small \inferrule*[right=WkId]{~}{\llparenthesis \Gamma \rrparenthesis \fatsemi \llparenthesis \Delta \rrparenthesis, x : {\uparrow}_{\llparenthesis \Gamma \rrparenthesis}^{\llparenthesis \Delta \rrparenthesis}(\llbracket S \rrbracket), \llparenthesis \Theta \rrparenthesis \vdash x \in \llbracket S \rrbracket\\ \gg \mathsf{proj}_{\llparenthesis \Gamma \rrparenthesis}^{\llparenthesis \Delta \rrparenthesis \mid \llbracket S \rrbracket} \circ \varoplus_{\llparenthesis \Delta \rrparenthesis}(\rho) \circ \varoplus_{\llparenthesis \Delta \rrparenthesis x : {\uparrow}^{\llparenthesis \Delta \rrparenthesis}_{\llparenthesis \Gamma \rrparenthesis}(\llbracket S \rrbracket)}(\top_{\llbracket \llparenthesis \Theta \rrparenthesis \rrbracket})}
$$ Similarly, if the ambient LFDC has contraction but \emph{not} term-level weakening, then we may make use of the following in addition to the strict identity rule above: $$
\small \inferrule*[right=StrContr]{~}{\llparenthesis \Gamma \rrparenthesis, x : \llbracket S \rrbracket \mid \epsilon \vdash x \in \omega_{\llbracket \llparenthesis \Gamma \rrparenthesis \rrbracket}^{\llbracket S \rrbracket}(\llbracket S \rrbracket) \gg \mathfrak{d}_{\llbracket \llparenthesis \Gamma \rrparenthesis \rrbracket, \llbracket S \rrbracket}}
$$ Finally, if the ambient LFDC is Cartesian, then we may instead use the following rule in addition to the weak idenity rule: $$
\small \inferrule*[right=WkContr]{~}{\llparenthesis \Gamma \rrparenthesis, x : \llbracket S \rrbracket, \llparenthesis \Delta \rrparenthesis \fatsemi \llparenthesis \Theta \rrparenthesis \vdash x \in \uparrow_{\llparenthesis \Gamma \rrparenthesis}^{x : \llbracket S \rrbracket, \llparenthesis \Delta \rrparenthesis}(\llbracket S \rrbracket)\\ \gg \uparrow_{\llparenthesis \Gamma \rrparenthesis, x : \llbracket S \rrbracket}^{\llparenthesis \Delta \rrparenthesis}(\mathfrak{d}_{\llbracket \llparenthesis \Gamma \rrparenthesis \rrbracket, \llbracket S \rrbracket}) \circ \top_{\llbracket \llparenthesis \Theta \rrparenthesis \rrbracket}}
$$

\subsubsection{Exchange} Unlike term-level weakening and contraction, exchange cannot be handled merely in the typing rules for individual variables, since it necessarily concerns the use of \emph{multiple} variables. Thus, the rules for exchange must allow permuting the variables in the term-level context of any typing judgment, i.e: $$
\small \inferrule*[right=ExchI]{\llparenthesis \Gamma \rrparenthesis \fatsemi \llparenthesis \Delta \rrparenthesis, y : \llbracket T \rrbracket, x : \omega_{\llbracket \llparenthesis \Gamma \rrparenthesis, \llparenthesis \Delta \rrparenthesis \rrbracket}^{\llbracket T \rrbracket}(\llbracket S \rrbracket), \mathsf{Exch}^*(\llparenthesis \Theta \rrparenthesis) \vdash \llbracket R \rrbracket \ni r \gg \llbracket r \rrbracket}{\llparenthesis \Gamma \rrparenthesis \fatsemi \llparenthesis \Delta \rrparenthesis, x : \llbracket S \rrbracket, y : \omega_{\llbracket \llparenthesis \Gamma \rrparenthesis, \llparenthesis \Delta \rrparenthesis \rrbracket}^{\llbracket S \rrbracket}(\llbracket T \rrbracket), \llparenthesis \Theta \rrparenthesis \vdash \llbracket R \rrbracket \ni r \gg \llbracket r \rrbracket \circ \mathsf{exch}}
$$ $$
\small \inferrule*[right=ExchE]{\llparenthesis \Gamma \rrparenthesis \fatsemi \llparenthesis \Delta \rrparenthesis, y : \llbracket T \rrbracket, x : \omega_{\llbracket \llparenthesis \Gamma \rrparenthesis, \llparenthesis \Delta \rrparenthesis \rrbracket}^{\llbracket T \rrbracket}(\llbracket S \rrbracket), \mathsf{Exch}^*(\llparenthesis \Theta \rrparenthesis) \vdash e \in \llbracket R \rrbracket \gg \llbracket e \rrbracket}{\llparenthesis \Gamma \rrparenthesis \fatsemi \llparenthesis \Delta \rrparenthesis, x : \llbracket S \rrbracket, y : \omega_{\llbracket \llparenthesis \Gamma \rrparenthesis, \llparenthesis \Delta \rrparenthesis \rrbracket}^{\llbracket S \rrbracket}(\llbracket T \rrbracket), \llparenthesis \Theta \rrparenthesis \vdash e \in \llbracket R \rrbracket \gg \llbracket e \rrbracket \circ \mathsf{exch}}
$$ where $$
\mathsf{Exch} = \alpha^{-1} \circ \llbracket \llparenthesis \Gamma \rrparenthesis, \llparenthesis \Delta \rrparenthesis \rrbracket _\bullet(\sigma)\circ \alpha \quad
\mathsf{exch} = \varoplus_{\llparenthesis \Delta \rrparenthesis}(\beta^{-1} \circ \langle \sigma, - \rangle \circ \beta)
$$ The rules for exchange are thus not syntax-directed, and so introduce an additional element of nondeterminism to type-checking, since type checking an expression may require trying all valid permutations of the context. For present purposes, this is a tolerable state of affairs, as there are only ever finitely many such permutations, and so this does not impact the decidability of type-checking. However, for practical use, further work will be necessary to find ways of cutting down on this nondeterminism. We can however omit these rules in Cartesian LFDCs, since by \textbf{Def. \ref{cart}} the exchange structure in a Cartesian LFDC already arises from term-level weakening and contraction.

\subsubsection{The Unit Type} The formation rule for the unit type is straightforward $$
\small \inferrule*[right=$\mathbbm{1}$Form]{~}{\llparenthesis \Gamma \rrparenthesis \vdash \mathbbm{1} ~ \mathsf{Type} \gg \mathbbm{1}_{\llbracket \llparenthesis \Gamma \rrparenthesis \rrbracket}}$$ Likewise the introduction rule, except that this rule admits a variation when the ambient LFDC has term-level weakening: $$\small \inferrule*[right=$\mathbbm{1}$IntroStr]{~}{\llparenthesis \Gamma \rrparenthesis \fatsemi \epsilon \vdash \mathbbm{1}_{\llbracket \llparenthesis \Gamma \rrparenthesis \rrbracket} \ni \langle \rangle \gg \text{id}_{\mathbbm{1}_{\llbracket \llparenthesis \Gamma \rrparenthesis \rrbracket}}}$$ $$ \small \inferrule*[right=$\mathbbm{1}$IntroWk]{~}{\llparenthesis \Gamma \rrparenthesis \fatsemi \llp{\Delta} \vdash \mathbbm{1}_{\llbracket \llparenthesis \Gamma \rrparenthesis \rrbracket} \ni \langle \rangle \gg \top_{\llb{\llp{\Delta}}}}
$$ However, the elimination rules for the unit type are fairly complex. There are two main reasons for this: one is that the unit type is \emph{positive} and so its elimination form is \emph{pattern-matching}, which may be applied either at the \emph{term level} or at the \emph{type level} -- hence there must be two distinct rules for such uses. Moreover, when performing such pattern-matching, we use some part $\llp{\Delta}$ of the ambient context in constructing a term $\llb{e}$ of the unit type, but the part of the context occurring after $\llp{\Delta}$ may implicitly depend upon $\llb{e}$ itself. To solve this issue, we require that a pattern matching expression represent these dependencies explicitly, by wrapping the remaining context in a type dependent upon the expression being matched over. We thus have the following term-level elimination rule: $$
\small \inferrule*[right=$\mathbbm{1}$Elim1]{
\llparenthesis \Gamma \rrparenthesis, \llparenthesis \Delta \rrparenthesis \fatsemi \llparenthesis \Theta \rrparenthesis \vdash e \in \mathbbm{1}_{\llbracket \llparenthesis \Gamma \rrparenthesis, \llparenthesis \Delta \rrparenthesis \rrbracket} \gg \llbracket e \rrbracket\\
\llparenthesis \Gamma \rrparenthesis \vdash R ~ \mathsf{Type} \gg \llbracket R \rrbracket\\ \llparenthesis \Gamma \rrparenthesis, \llparenthesis \Delta \rrparenthesis, a : \mathbbm{1}_{\llbracket \llparenthesis \Gamma \rrparenthesis, \llparenthesis \Delta \rrparenthesis \rrbracket} \vdash U ~ \mathsf{Type} \gg \llbracket U \rrbracket\\ \llparenthesis \Gamma \rrparenthesis, \llparenthesis \Delta \rrparenthesis, \llparenthesis \Theta \rrparenthesis \fatsemi \llparenthesis \Phi \rrparenthesis \vdash {\downarrow}(\llbracket e \rrbracket)^*(\llbracket U \rrbracket) \ni u \gg \llbracket u \rrbracket\\ \llparenthesis \Gamma \rrparenthesis \fatsemi \llparenthesis \Delta \rrparenthesis, z : \eta^*(\llbracket U \rrbracket) \vdash \llbracket R \rrbracket \ni r \gg \llbracket r \rrbracket}{\llparenthesis \Gamma \rrparenthesis \fatsemi \llparenthesis \Delta \rrparenthesis, \llparenthesis \Theta \rrparenthesis, \llparenthesis \Phi \rrparenthesis \vdash \mathsf{let}[a ._R^U] \langle \rangle = e ~ \mathsf{with} ~ z = u ~ \mathsf{in} ~ r \in \llbracket R \rrbracket\\ \gg \llbracket r \rrbracket \circ \varoplus_{\llparenthesis \Delta \rrparenthesis}(\ell \circ \langle \llbracket e \rrbracket, - \rangle) \circ \varoplus_{\llparenthesis \Delta \rrparenthesis, \llparenthesis \Theta \rrparenthesis}(\llbracket u \rrbracket)}
$$ and the following type-level elimination rule: $$
\small
\inferrule*[right=$\mathbbm{1}$Elim2]{\llp{\Gamma} \fatsemi \llp{\Delta} \vdash e \in \mathbbm{1}_{\llb{\llp{\Gamma}}} \gg \llb{e}\\
\llp{\Gamma}, a : \mathbbm{1}_{\llb{\llp{\Gamma}}} \vdash R ~ \mathsf{Type} \gg \llb{R}\\ \llp{\Gamma}, b : \mathbbm{1}_{\llb{\llp{\Gamma}}} \vdash U ~ \mathsf{Type} \gg \llb{U}\\ \llp{\Gamma}, b : \mathbbm{1}_{\llb{\llp{\Gamma}}}, c : \llb{U} \vdash V ~ \mathsf{Type} \gg \llb{V}\\ \llp{\Gamma}, \llp{\Delta} \fatsemi \llp{\Theta} \vdash {\downarrow}(\llb{e})^*(\llb{U}) \ni u \gg \llb{u}\\ \llp{\Gamma}, \llp{\Delta}, \llp{\Theta} \fatsemi \llp{\Phi} \vdash {\downarrow}(\llb{u})^*({\downarrow}(\llb{e})_\bullet(\llb{U})^*(\llb{V})) \ni v \gg \llb{v}\\ \llp{\Gamma}, w : \eta^*(\llb{U}) \fatsemi z : \eta_\bullet(\llb{U})^*(\llb{V}) \vdash \eta^*(\llb{R}) \ni r \gg \llb{r}}{\llp{\Gamma}, \llp{\Delta}, \llp{\Theta} \fatsemi \llp{\Phi} \vdash\\ \mathsf{let}[a,b,c._R^{U,V}]\langle \rangle = e ~ \mathsf{with} ~ w = u ~ \mathsf{and} ~ z = v ~ \mathsf{in} ~ r \in {\downarrow}(\llb{e})^*(\llb{R}) \\ \gg {\downarrow}(\llb{u})^*({\downarrow}(\llb{e})_\bullet(\llb{U})^*(\eta^{-1}_\bullet(\eta^*(\llb{U}))^*(\llb{r}))) \circ \llb{v}}
$$ On their own these rules are not quite sufficient, due to a quirk of the unit type: if the part of a context ocurring after that used in constructing $\llb{e}$ is empty, we cannot encode this via a variable of the unit type, since we would then generally need to eliminate this variable, bringing us round in a circle. We thus have the following additional rules to handle these exceptional cases at the term level: $$
\small \inferrule*[right=$\mathbbm{1}$Elim3]{\llparenthesis \Gamma \rrparenthesis , \llparenthesis \Delta \rrparenthesis \fatsemi \llparenthesis \Theta \rrparenthesis \vdash e \in \mathbbm{1}_{\llbracket \llparenthesis \Gamma \rrparenthesis, \llparenthesis \Delta \rrparenthesis \rrbracket} \gg \llbracket e \rrbracket\\
\llparenthesis \Gamma \rrparenthesis \vdash R ~ \mathsf{Type} \gg \llbracket R \rrbracket\\ \llparenthesis \Gamma \rrparenthesis \fatsemi \llparenthesis \Delta \rrparenthesis \vdash \llbracket R \rrbracket \ni r \gg \llbracket r \rrbracket}{\llparenthesis \Gamma \rrparenthesis \fatsemi \llparenthesis \Delta \rrparenthesis, \llparenthesis \Theta \rrparenthesis \vdash \mathsf{let}[_R] \langle \rangle = e ~ \mathsf{in} ~ r \in \llbracket R \rrbracket \gg \llbracket r \rrbracket \circ \varoplus_{\llparenthesis \Delta \rrparenthesis}(\rho \circ \llbracket e \rrbracket)}
$$ and at the type level: $$ \small
\inferrule*[right=$\mathbbm{1}$Elim4]{\llp{\Gamma} \fatsemi \llp{\Delta} \vdash e \in \mathbbm{1}_{\llb{\llp{\Gamma}}} \gg \llb{e}\\ 
\llp{\Gamma}, a : \mathbbm{1}_{\llb{\llp{\Gamma}}} \vdash R ~ \mathsf{Type} \gg \llb{R}\\ \llp{\Gamma}, b : \mathbbm{1}_{\llb{\llp{\Gamma}}} \vdash U ~ \mathsf{Type} \gg \llb{U}\\ \llp{\Gamma}, \llp{\Delta} \fatsemi \llp{\Theta} \vdash {\downarrow}(\llb{e})^*(\llb{U}) \ni u \gg \llb{u}\\ \llp{\Gamma} \fatsemi z : \eta^*(\llb{U}) \vdash \eta^*(\llb{R}) \ni r \gg \llb{r}}{\llp{\Gamma}, \llp{\Delta} \fatsemi \llp{\Theta} \vdash \mathsf{let}[a,b._R^U] \langle \rangle = e ~ \mathsf{with} ~ z = u ~ \mathsf{in} ~ r\\ \in {\downarrow}(\llb{e})^*(\llb{R}) \gg {\downarrow}(\llb{e})^*(\eta^*(\llb{r})) \circ \llb{u}}
$$

\subsubsection{Dependent Pair Types} The formation rule for dependent pair types is straightforward: $$
\small \inferrule*[right=$\varoplus$Form]{\llp{\Gamma} \vdash S ~ \mathsf{Type} \gg \llb{S}\\ \llp{\Gamma}, x : \llb{S} \vdash T ~ \mathsf{Type} \gg \llb{T}}{\llp{\Gamma} \vdash \bigoplus x : S . T ~ \mathsf{Type \gg \varoplus_{\llb{S}}}(\llb{T})}
$$ As is the introduction rule: $$
\small \inferrule*[right=$\varoplus$Intro]{\llp{\Gamma} \fatsemi \llp{\Delta} \vdash \llb{S} \ni s \gg \llb{s}\\ \llp{\Gamma}, \llp{\Delta} \fatsemi \llp{\Theta} \vdash {\downarrow}(\llb{s})^*(\llb{T}) \ni t \gg \llb{t}}{\llp{\Gamma} \fatsemi \llp{\Delta}, \llp{\Theta} \vdash \varoplus_{\llb{S}}(\llb{T}) \ni \langle s, t \rangle\\ \gg \langle \llb{s}, - \rangle \circ {\oplus}\mathsf{asc}^{\llp{\Delta} \mid \llb{\llp{\Theta}}}_{\llp{\Gamma}} \circ \varoplus_{\llp{\Delta}}(\llb{t})}
$$ Because the dependent pair type is positive, its elimination rules follow the same pattern as the unit type. We have a term-level rule: $$ \small
\inferrule*[right=$\varoplus$Elim1]{\llp{\Gamma}, \llp{\Delta} \fatsemi \llp{\Theta} \vdash e \in \varoplus_{\llb{S}}(\llb{T}) \gg \llb{e}\\
\llp{\Gamma} \vdash R ~ \mathsf{Type} \gg \llb{R} \\ \llp{\Gamma}, \llp{\Delta}, a : \varoplus_{\llb{S}}(\llb{T}) \vdash U ~ \mathsf{Type} \gg \llb{U}\\ \llp{\Gamma}, \llp{\Delta}, \llp{\Theta} \fatsemi \llp{\Phi} \vdash {\downarrow}(\llb{e})^*(\llb{U}) \ni u \gg \llb{u}\\ \llp{\Gamma} \fatsemi \llp{\Delta}, x : \llb{S}, y : \llb{T}, z : (\alpha^{-1})^*(\llb{U}) \vdash \llb{R} \ni r \gg \llb{r}}{\llp{\Gamma} \fatsemi \llp{\Delta}, \llp{\Theta}, \llp{\Phi} \vdash \mathsf{let}[a._R^U] \langle x,y \rangle = e ~ \mathsf{with} ~ z = u ~ \mathsf{in} ~ r \in \llb{R}\\ \gg \llb{r} \circ \varoplus_{\llp{\Delta}}(\beta^{-1} \circ \langle \llb{e}, - \rangle \circ {\oplus}\mathsf{asc}_{\llp{\Gamma}, \llp{\Delta}}^{\llp{\Theta} \mid {\downarrow}(\llb{e})^*(\llb{U})} \circ \varoplus_{\llp{\Theta}}(\llb{u}))}
$$ and a type-level rule: $$ \small
\inferrule*[right=$\varoplus$Elim2]{\llp{\Gamma} \fatsemi \llp{\Delta} \vdash e \in \oplus_{\llb{S}}(\llb{T}) \gg \llb{e}\\ \llp{\Gamma}, a : \varoplus_{\llb{S}}(\llb{T}) \vdash R ~ \mathsf{Type} \gg \llb{R}\\ \llp{\Gamma}, b : \varoplus_{\llb{S}}(\llb{T}) \vdash U ~ \mathsf{Type} \gg \llb{U}\\ \llp{\Gamma}, b : \varoplus_{\llb{S}}(\llb{T}), c : \llb{U} \vdash V ~ \mathsf{Type} \gg \llb{V}\\ \llp{\Gamma}, \llp{\Delta} \fatsemi \llp{\Theta} \vdash {\downarrow}(\llb{e})^*(\llb{U}) \ni u \gg \llb{u}\\ \llp{\Gamma}, \llp{\Delta}, \llp{\Theta} \fatsemi \llp{\Phi} \vdash {\downarrow}(\llb{u})^*({\downarrow}(\llb{e})_\bullet(\llb{U})^*(\llb{V})) \ni v \gg \llb{v}\\ \llp{\Gamma}, x : \llb{S}, y : \llb{T}, w : (\alpha^{-1})^*(\llb{U}) \fatsemi z : \alpha_\bullet^{-1}(\llb{U})^*(\llb{V}) \vdash\\ \llb{R} \ni r \gg \llb{r}}{\llp{\Gamma}, \llp{\Delta}, \llp{\Theta} \fatsemi \llp{\Phi} \vdash\\ \mathsf{let}[a,b,c._R^{U,V}] \langle x,y \rangle = e ~ \mathsf{with} ~ w = u ~ \mathsf{and} ~ z = v ~ \mathsf{in} ~ r \in {\downarrow}(\llb{e})^*(\llb{R})\\ \gg {\downarrow}(\llb{u})^*({\downarrow}(\llb{e})_\bullet(\llb{U})^*(\alpha_\bullet((\alpha^{-1})^*(\llb{U}))^*(\llb{r}))) \circ  \llb{v}}
$$

\subsubsection{Term-Level Function Types} The rules for term-level functions are each essentially variations on the rules for function types in intuitionistic type theory. We have the formation rule for $\sslash$: $$
\small \inferrule*[right=$\sslash$Form]{\llp{\Gamma} \vdash S ~ \mathsf{Type} \gg \llbracket S \rrbracket\\ \llp{\Gamma} \vdash T ~ \mathsf{Type} \gg \llbracket T \rrbracket}{\llparenthesis \Gamma \rrparenthesis \vdash S \sslash T ~ \mathsf{Type} \gg \sslash_{\llbracket S \rrbracket}(\llbracket T \rrbracket)}
$$ and the following for $\bbslash$: $$
\small \inferrule*[right=$\bbslash$Form]{\llp{\Gamma} \vdash S ~ \mathsf{Type} \gg \llbracket S \rrbracket\\ \llp{\Gamma} \vdash T ~ \mathsf{Type} \gg \llbracket T \rrbracket}{\llp{\Gamma} \vdash S \bbslash T ~ \mathsf{Type} \gg \bbslash_{\llbracket S \rrbracket}(\llbracket T \rrbracket)}
$$ The intro rule for $\sslash$ follows the usual form of function abstraction, introducing a variable on the \emph{left} of the term-level context: $$
\small \inferrule*[right=$\sslash$Intro]{\llparenthesis \Gamma \rrparenthesis \fatsemi x : \llbracket S \rrbracket, \omega_{\llbracket \llparenthesis \Gamma \rrparenthesis \rrbracket}^{\llbracket S \rrbracket}(\llparenthesis \Delta \rrparenthesis) \vdash \llbracket T \rrbracket \ni t \gg \llbracket t \rrbracket}{\llparenthesis \Gamma \rrparenthesis \fatsemi \llparenthesis \Delta \rrparenthesis \vdash \sslash_{\llbracket S \rrbracket}(\llbracket T \rrbracket) \ni \sslash x . t \gg \sslash(\llbracket t \rrbracket)}
$$ while the intro rule for $\bbslash$ introduces a variable on \emph{right}: $$
\small \inferrule*[right=$\bbslash$Intro]{\llparenthesis \Gamma \rrparenthesis \fatsemi \llparenthesis \Delta \rrparenthesis, x : \uparrow_{\llp{\Gamma}}^{\llp{\Delta}}(\llb{S}) \vdash \llbracket T \rrbracket \ni t \gg \llbracket t \rrbracket}{\llparenthesis \Gamma \rrparenthesis \fatsemi \llparenthesis \Delta \rrparenthesis \vdash \bbslash_{\llbracket S \rrbracket}(\llbracket T \rrbracket) \ni \bbslash x . t \gg \bbslash(\llbracket t \rrbracket)}
$$ The elimination rule for $\sslash$ thus forms its argument using part of the context occurring to the \emph{left} of the part used in forming a function: $$
\small \inferrule*[right=$\sslash$Elim]{\llparenthesis \Gamma \rrparenthesis \fatsemi \llparenthesis \Theta \rrparenthesis \vdash f \in \sslash_{\llbracket S \rrbracket}(\llbracket T \rrbracket) \gg \llbracket f \rrbracket\\ \llparenthesis \Gamma \rrparenthesis \fatsemi \llparenthesis \Delta \rrparenthesis \vdash \llbracket S \rrbracket \ni s \gg \llbracket s \rrbracket}{\llparenthesis \Gamma \rrparenthesis \fatsemi \llparenthesis \Delta \rrparenthesis, {\uparrow_{\llparenthesis \Gamma \rrparenthesis}^{\llparenthesis \Delta \rrparenthesis}}(\llparenthesis \Theta \rrparenthesis) \vdash s \triangleright f \in \llbracket T \rrbracket\\ \gg \sslash^{-1}(\llb{f}) \circ \langle \llb{s}, - \rangle \circ {\oplus}\mathsf{asc}^{\llp{\Delta} \mid \llb{\uparrow_{\llp{\Gamma}}^{\llp{\Delta}}(\llp{\Theta})}}_{\llp{\Gamma}}}
$$ while the elimination rule for $\bbslash$ forms the argument to $f$ using part of the context occurring to the \emph{right} of the part used in forming $f$: $$
\small \inferrule*[right=$\bbslash$Elim]{\llparenthesis \Gamma \rrparenthesis \fatsemi \llparenthesis \Delta \rrparenthesis \vdash f \in \bbslash_{\llbracket S \rrbracket}(\llbracket T \rrbracket) \gg \llbracket f \rrbracket\\ \llparenthesis \Gamma \rrparenthesis \fatsemi \llparenthesis \Theta \rrparenthesis \vdash \llbracket S \rrbracket \ni s \gg \llbracket s \rrbracket}{\llparenthesis \Gamma \rrparenthesis \fatsemi \llparenthesis \Delta \rrparenthesis, {\uparrow_{\llbracket \llparenthesis \Gamma \rrparenthesis \rrbracket}^{\llparenthesis \Delta \rrparenthesis}}(\llparenthesis \Theta \rrparenthesis) \vdash f \triangleleft s \in \llbracket T \rrbracket \gg \bbslash^{-1}(\llb{f}) \circ \uparrow_{\llp{\Gamma}}^{\llp{\Delta}}(\llb{s})}
$$

\subsubsection{Type-Level Function Types} The rules for type-level functions are again a variation on the rules for function types in intuitionistic type theory, but this time the function types in question are \emph{dependent function types}. We have the following formation rule: $$
\small \inferrule*[right=$\forall$Form]{\llp{\Gamma} \vdash S ~ \mathsf{Type} \gg \llb{S}\\ \llp{\Gamma}, x : \llb{S} \vdash T ~ \mathsf{Type} \gg \llb{T}}{\llp{\Gamma} \vdash \forall x : S . T ~ \mathsf{Type} \gg \forall_{\llb{S}}(\llb{T})}
$$ and the following introduction rule $$
\small \inferrule*[right=$\forall$Intro]{\llp{\Gamma}, x : \llb{S} \fatsemi \omega_{\llb{\llp{\Gamma}}}^{\llb{S}}(\llp{\Delta}) \vdash \llb{T} \ni t \gg \llb{t}}{\llp{\Gamma} \fatsemi \llp{\Delta} \vdash \forall_{\llb{S}}(\llb{T}) \ni \Lambda x . t \gg \Lambda(\llb{t})}
$$ and the corresponding elimination rule $$
\small \inferrule*[right=$\forall$Elim]{\llp{\Gamma} \fatsemi \llp{\Theta} \vdash f \in \forall_{\llb{S}}(\llb{T}) \gg \llb{f}\\ \llp{\Gamma} \fatsemi \llp{\Delta} \vdash \llb{S} \ni s \gg \llb{s}}{\llp{\Gamma}, \llp{\Delta} \fatsemi \uparrow_{\llp{\Gamma}}^{\llp{\Delta}}(\llp{\Theta}) \vdash f \cdot s \in {\downarrow}(\llb{s})^*(\llb{T})\\ \gg {\downarrow}(\llb{s})^*(\Lambda^{-1}(\llb{f}))}
$$ wherein we may make use of any part of the type-level context in forming an input to $f$, provided that the term-level context does not depend upon this part of the type-level context.

\subsubsection{Product Types} The rules for product types are largely straightforward. We have the following type formation rule: $$
\inferrule*[right=$\times$Form]{\llp{\Gamma} \vdash S ~ \mathsf{Type} \gg \llb{S}\\ \llp{\Gamma} \vdash T ~ \mathsf{Type } \gg \llb{T}}{\llp{\Gamma} \vdash S \times T ~ \mathsf{Type} \gg \llb{S} \times_{\llb{\llp{\Gamma}}} \llb{T}}
$$ and the following introduction rule $$
\inferrule*[right=$\times$Intro]{\llp{\Gamma} \fatsemi \llp{\Delta} \vdash \llb{S} \ni s \gg \llb{s}\\ \llp{\Gamma} \fatsemi \llp{\Delta} \vdash \llb{T} \ni t \gg \llb{t}}{\llp{\Gamma} \fatsemi \llp{\Delta} \vdash \llb{S} \times_{\llb{\llp{\Gamma}}} \llb{T} \ni (s,t) \gg (\llb{s}, \llb{t})}
$$ Note that the term-level context $\llp{\Delta}$ is used in checking both $s$ and $t$, since $(s,t)$ offers a \emph{choice} of which of $s,t$ to construct from the resources in $\llp{\Delta}$. We then have the following elimination rules, which allow for making such a choice: $$
\small \inferrule*[right=$\times$Elim1]{\llp{\Gamma} \fatsemi \llp{\Delta} \vdash e \in \llb{S} \times_{\llb{\llp{\Gamma}}} \llb{T} \gg \llb{e}}{\llp{\Gamma} \fatsemi \llp{\Delta} \vdash \pi_1(e) \in \llb{S} \gg \pi_1 \circ \llb{e}}$$ and $$ \small \inferrule*[right=$\times$Elim2]{\llp{\Gamma} \fatsemi \llp{\Delta} \vdash e \in \llb{S} \times_{\llb{\llp{\Gamma}}} \llb{T} \gg \llb{e}}{\llp{\Gamma} \fatsemi \llp{\Delta} \vdash \pi_2(e) \in \llb{T} \gg \pi_2 \circ \llb{e}}
$$

\subsection{Syntactic Completeness}

As a consequence of the construction of the rules given above, we automatically have a form of type soundness for the theory: every well typed syntactic expression corresponds to a well defined semantic object of the appropriate kind. Beyond such soundness, however, there are further desiderata we may have for such a theory, namely \emph{completeness} and \emph{effectivity/decidability} of the above-defined procedure for type-checking/computing denotations of expressions.

As to the \emph{completeness} of this theory, by soundness we already have that the syntax of the theory may be interpreted in any strict LFDC, so it remains only to show that this syntax is closed under the constructions available in a (strict) LFDC, and therefore that the syntax itself forms such a (strict) LFDC. From this it will follow that a syntactic expression is well typed if and only if a corresponding semantic object of the appropriate kind exists in \emph{every} strict LFDC.

The syntax of the theory already includes primitive constructs corresponding to all the semantic type-formers in an LFDC with type-level weakening, function types, and product types. Therefore all that remains is to prove the \emph{admissibility} of syntactic constructions corresponding to the parts of such an LFDC not given by its type-formers, which are namely: type-level weakening, and substitution/composition.

\begin{theorem} ~
\begin{itemize}
\item If $\llp{\Gamma}, \llp{\Delta} \vdash T ~ \mathsf{Type} \gg \llb{T}$ then $$\llp{\Gamma}, x : \llb{S}, \omega_{\llb{\llp{\Gamma}}}^{\llb{S}}(\llp{\Delta}) \vdash T ~ \mathsf{Type} \gg \mathsf{W}_{\llp{\Gamma}}^{\llb{S} \mid \llp{\Delta}}(\llb{T})$$
\item If $\llp{\Gamma}, \llp{\Delta} \fatsemi \llp{\Theta} \vdash \llb{T} \ni t \gg \llb{t}$ then $$\begin{array}{c} \llp{\Gamma}, x : \llb{S}, \omega_{\llb{\llp{\Gamma}}}^{\llb{S}}(\llp{\Delta}) \fatsemi \mathsf{W}_{\llp{\Gamma}}^{\llb{S} \mid \llp{\Delta}}(\llp{\Theta}) \vdash\\ \mathsf{W}_{\llp{\Gamma}}^{\llb{S} \mid \llp{\Delta}}(\llb{T}) \ni t \gg \mathsf{W}_{\llp{\Gamma}}^{\llb{S} \mid \llp{\Delta}}(\llb{t}) \end{array}$$
\item If $\llp{\Gamma}, \llp{\Delta} \fatsemi \llp{\Theta} \vdash e \in \llb{R} \gg \llb{e}$ then $$\begin{array}{c} \llp{\Gamma}, x : \llb{S}, \omega_{\llb{\llp{\Gamma}}}^{\llb{S}}(\llp{\Delta}) \fatsemi \mathsf{W}_{\llp{\Gamma}}^{\llb{S} \mid \llp{\Delta}}(\llp{\Theta}) \vdash\\ e \in \mathsf{W}_{\llp{\Gamma}}^{\llb{S} \mid \llp{\Delta}}(\llb{R}) \gg \mathsf{W}_{\llp{\Gamma}}^{\llb{S} \mid \llp{\Delta}}(\llb{e}) \end{array}$$
\end{itemize}
\end{theorem}

\begin{proof} Induction on derivations. \end{proof}

\begin{theorem}
Given expressions $T,t,f$ and an expression $e$ not containing any variable bound in $T,t,f$, respectively, write $T[e/x], t[e/x]$, and $f[e/x]$ for the uniform substitution of $e$ for all free occurrences of the variable $x$ in $T,t,f$, respectively. We then have the following:

\begin{itemize}
\item If $\llp{\Gamma}, x : \llb{S}, \llp{\Theta} \vdash T ~ \mathsf{Type} \gg \llb{T}$ and $\llp{\Gamma} \fatsemi \llp{\Delta} \vdash e \in \llb{S} \gg \llb{e}$ then $$\llp{\Gamma}, \llp{\Delta}, \llb{e}^*(\llp{\Theta}) \vdash T[e/x] \gg {\Downarrow}(\llb{e})^*(\llb{T})$$
\item If $\llp{\Gamma} \fatsemi \llp{\Delta}, x : \llb{S}, \llp{\Phi} \vdash \llb{T} \ni t \gg \llb{t}$ and $\llp{\Gamma}, \llp{\Delta} \fatsemi \llp{\Theta} \vdash e \in \llb{S} \gg \llb{e}$, then $$\begin{array}{c} \llp{\Gamma} \fatsemi \llp{\Delta}, \llp{\Theta}, \llb{e}^*(\llp{\Phi}) \vdash\\ \llb{T} \ni t[e/x] \gg \llb{t} \circ \varoplus_{\llp{\Delta}}(\llb{e}) \end{array}$$
\item If $\llp{\Gamma}, x : \llb{S}, \llp{\Theta} \fatsemi \llp{\Phi} \vdash \llb{T} \ni t \gg \llb{t}$ and $\llp{\Gamma} \fatsemi \llp{\Delta} \vdash e \in \llb{S} \gg \llb{e}$, then $$\begin{array}{c} \llp{\Gamma}, \llp{\Delta}, \llb{e}^*(\llp{\Theta}) \fatsemi {\Downarrow}(\llb{e})^*(\llp{\Phi}) \vdash\\ {\Downarrow}(\llb{e})^*(\llb{T}) \ni t[e/x] \gg {\Downarrow}(\llb{e})^*(\llb{t}) \end{array}$$
\item If $\llp{\Gamma} \fatsemi \llp{\Delta}, x : \llb{S}, \llp{\Phi} \vdash f \in \llb{R} \gg \llb{f}$ and $\llp{\Gamma}, \llp{\Delta} \fatsemi \llp{\Theta} \vdash e \in \llb{S} \gg \llb{e}$, then $$\begin{array}{c}\llp{\Gamma} \fatsemi \llp{\Delta}, \llp{\Theta}, \llb{e}^*(\llp{\Phi}) \vdash\\ f[e/x] \in \llb{R} \gg \llb{f} \circ \varoplus_{\llp{\Delta}}(\llb{e}) \end{array}$$
\item If $\llp{\Gamma}, x : \llb{S}, \llp{\Theta} \fatsemi \llp{\Phi} \vdash f \in \llb{R} \gg \llb{f}$ and $\llp{\Gamma} \fatsemi \llp{\Delta} \vdash e \in \llb{S} \gg \llb{e}$, then $$\begin{array}{c}\llp{\Gamma}, \llp{\Delta}, \llb{e}^*(\llp{\Theta}) \fatsemi {\Downarrow}(\llb{e})^*(\llp{\Phi}) \vdash\\ f \in {\Downarrow}(\llb{e})^*(\llb{T}) \gg {\Downarrow}(\llb{e})^*(\llb{f}) \end{array}$$
\end{itemize}
\end{theorem}

\begin{proof}
Induction on derivations.
\end{proof}

I leave to future work a full proof that the well typed syntactic fragment of this type theory, quotiented by judgmental equality as defined in \textbf{Def. \ref{judgmental}}, forms a strict LFDC with type-level weakening, function types, products, and the appropriate structural properties. Suffice it to say, however, that the above two propositions form the backbone of such a proof, and moreover demonstrate at least morally that the syntax of this type theory completely captures the type-theoretic language of such LFDCs.

\subsection{Substructuralization \& Decidability}

As to the \emph{decidability} of the described type theory, in general, one should not anticipate decidability when interpreting the type theory in an arbitrary (strict) LFDC. One may hope, however, to isolate a computationally well-behaved subclass of (strict) LFDCs, ensuring decidability for the associated type theories. For this purpose, I define a notion of \emph{substructuralization} that allows to convert an intuitionistic dependent type theory into a substructural one.

\begin{dfn}
Let $\mathcal{T}$ be an intuitionistic dependent type theory with \emph{judgmentally-distinct} type formers $\top, \Sigma, \Pi, \leftarrow, \rightarrow, \times$ such that \begin{enumerate}
\item $\top$ satisfies the rules of a unit type in intuitionistic type theory
\item $\Sigma$ satisfies the rules of a dependent pair type former in intuitionistic type theory
\item $\Pi$ satisfies the rules of a dependent function type former in intuitionistic type theory
\item $\to$ and $\leftarrow$ both satisfy the rules of function type formers in intuitionsitic type theory
\item $\times$ satisfies the rules of a product type former in intuitionistic type theory
\end{enumerate} Then the \emph{substructuralization} $\mathfrak{S}(\mathcal{T})$ of $\mathcal{T}$ is defined as the interpretation of the substructural dependent type theory defined above (with any combination of weakening, contraction, and exchange) in the strict LFDC with type-level weakening, function types, and products, given by the syntactic model of $\mathcal{T}$ as defined in \textbf{Ex. \ref{synmodels}}.
\end{dfn}

The substructuralization of an intuitionistic dependent type theory $\mathcal{T}$ essentially inherits its computational procedures from $\mathcal{T}$ while imposing substructural constraints upon the typing rules for $\mathcal{T}$. The essential idea behind this is that the \emph{terms} of substructural type theory denote the same sorts of data as those of intuitionistic type theory, i.e. functions, pairs, etc., whose computational behavior is already well understood and unchanged by the substructural rules. Hence we should be able to bootstrap ourselves up from an intuitionistic dependent type theory to a substructural dependent type theory, whilst preserving all the desirable computational properties thereof. To this effect, we have the following theorem:

\begin{theorem}
Let $\mathcal{T}$ be an intuitionistic dependent type theory that is \emph{normalizing} (i.e. every term of $\mathcal{T}$ computes to a judgmentally-equal normal form), and that has \emph{Type-Canonicity} in that: \begin{itemize}
\item if $R$ is judgmentally equal to $\top$, then its normal form is $\top$,
\item if $R$ is judgmentally equal to $\Sigma_S(T)$, then its normal form is $\Sigma_{S'}(T')$ for some types $S', T'$ in normal form,
\item etc.
\end{itemize} then $\mathfrak{S}(\mathcal{T})$ has the following properties: \begin{enumerate}
\item Judgmental equality of types is decidable. I.e. given semantic types $\llb{S}$ and $\llb{T}$, we may check whether these are equal in the syntactic model of $\mathcal{T}$ by reducing them both to normal form and comparing these normal forms for $\alpha$-equivalence.
\item Pattern-matching on types is decidable. E.g. to check whether a type is of the form $\varoplus_{\llb{S}}(\llb{T})$, we reduce it to normal form and check whether this normal form has the form $\Sigma_{S'}(T')$.
\end{enumerate}
\end{theorem}

\begin{corollary}
If an intuitionistic dependent type theory $\mathcal{T}$ satisfies the conditions of the above theorem, then type checking for $\mathfrak{S}(\mathcal{T})$ is decidable. Examining the rules for $\mathfrak{S}(\mathcal{T})$, we see that these require only the abilities to 1) pattern match on expressions (trivial), 2) pattern match on types using primitive type formers (follows from the above theorem), 3) pattern match for type-level weakening (can be done by checking that the weakened variable does not occur freely in a type), and 4) check types for equality (follows from the above theorem).
\end{corollary}

\subsection{Application: Linear Logical Frameworks}

I come now to an example of the practical advantage of this theory over prior substructural dependent type theories, namely: a suitably substructuralized dependent type theory is particularly well-suited as a logical framework for the metatheory of \emph{linear logic}.

\begin{dfn}
We define the \emph{Logic of Left-Fibred Double Categories with Symmetry} (LLFDC$^\sigma$), as the substructuralization of intuitionistic dependent type theory with at least one universe, which includes the exchange rules but neither term-level weakening nor contraction as structural rules.
\end{dfn}

\begin{dfn} We may add atomic type families and axioms / constants to LLFDC$^\sigma$ by including corresponding variables of the appropriate types in the underlying intuitionistic dependent type theory and then adding these to the signature of its substructuralization. Hence for an atomic type family variable $P$ and a closed LLFDC$^\sigma$-type $S$, we write $P(S) \in \Sigma$ to mean $P(\llb{S}) \gg P \in \Sigma$, and similarly for a constant variable $p$, we write $p : S \in \Sigma$ to mean $p : \llb{S} \gg p \in \Sigma$.
\end{dfn}

I write $\forall x_1, \dots, x_n : S . T$ and $\bigoplus x_1, \dots, x_n  : S . T$ as abbreviations for $\forall x_1 : S . \dots \forall x_n : S . T$ and $\bigoplus x_1 : S . \dots \bigoplus x_n : S . T$, respectively. Similarly, I write $\langle t_1, \dots, t_n \rangle$ in place of $\langle t_1, \langle t_2, \dots \langle t_{n-1}, t_n \rangle \dots \rangle \rangle$. Additionally, I shall write $S \varotimes T$ for $\bigoplus x : S . T$ when $x$ does not occur free in $T$, and $S \multimap T$ instead of $T \bbslash S$. I also generally write $\underline{e}$ simply as $e$, i.e. I treat embedding of elimination forms into introduction forms as an implicit operation, rather than an explicit one.

I claim that LLFDC$^\sigma$ is an ideal setting in which to conduct the metatheory of ordinary (intuitionistic) linear logic, as I shall now demonstrate. I will show, in particular, that LLFDC$^\sigma$ is capable of representing cut admissibility for intuitionistic linear sequent calculus in a manner which avoids the problems with such representations in prior linear dependent type theories highlighted by Reed \cite{reed}.

I follow the method of Cervesato \& Pfenning \cite{cervesato-pfenning}, as adapted by Reed \cite{reed}, in representing linear sequent calculus via HOAS, with suitable modifications for the specificities of LLFDC$^\sigma$. We begin by postualting an atomic type of \emph{propositions} $$
\mathsf{Prop}(\mathbbm{1}) \in \Sigma
$$ which is then used to parameterize atomic type families of \emph{antecedents} and \emph{consequents}, respectively: $$
\mathsf{Ante}(\mathsf{Prop}\langle \rangle) \in \Sigma \qquad \mathsf{Conse}(\mathsf{Prop}\langle \rangle) \in \Sigma
$$ the idea being that a derivation $a_1, \dots, a_n \vdash c$ in intuitionistic linear sequent calculus corresponds to a closed term of type $$\forall a_1, \dots, a_n, c : \mathsf{Prop} \langle \rangle . (\mathsf{Ante}(a_1) \varotimes \dots \varotimes \mathsf{Ante}(a_n)) \multimap \mathsf{Conse}(c)$$ We then include constructors on propositions corresponding to the connectives of linear logic. For illustrative purposes, I concentrate on the linear implication $\multimap$, represented as follows: $$
{\multimap} : \mathsf{Prop} \langle \rangle \multimap \mathsf{Prop} \langle \rangle \multimap \mathsf{Prop} \langle \rangle \in \Sigma
$$ For the sake of legibility, I write and $a \multimap b$ as an abbreviation for $\multimap \triangleleft a \triangleleft b$. We then have the following constructors for derivations, corresponding to the left/right rules of each of implication in linear sequent calculus: $$
\begin{array}{rcll}
{\multimap} \mathsf{R} & : & \forall a,b : \mathsf{Prop}\langle \rangle .\\ & & (\mathsf{Ante}(a) \multimap \mathsf{Conse}(b)) \multimap \mathsf{Conse}(a \multimap b) & \in \Sigma\\
{\multimap} \mathsf{L} & : & \forall a,b,c : \mathsf{Prop}\langle \rangle . \mathsf{Conse}(a)\\ & & \multimap (\mathsf{Ante}(b) \multimap \mathsf{Conse}(c))\\ & & \multimap \mathsf{Ante}(a \multimap b) \multimap \mathsf{Conse}(c) & \in \Sigma
\end{array}$$ Additionally, we have the following constructors, corresponding to the Cut and Identity rules: $$
\begin{array}{rcl}
\mathsf{id} & : & \forall a : \mathsf{Prop}\langle \rangle . \mathsf{Ante}(a) \multimap \mathsf{Conse}(a) \in \Sigma\\
\mathsf{cut} & : & \forall a,b : \mathsf{Prop} \langle \rangle .\\ & & (\mathsf{Conse}(a) \otimes (\mathsf{Ante}(a) \multimap \mathsf{Conse}(b))) \multimap \mathsf{Conse}(b) \in \Sigma
\end{array}
$$ Our goal, then, is to give a procedure for converting a derivation making use of $\mathsf{cut}$ into a \emph{cut-free} derivation. For this purpose, we introduce a type family representing constructions of cut-free proofs: $$ 
\mathsf{CF}(\bigoplus a : \mathsf{Prop} \langle \rangle . \mathsf{Conse}(a)) \in \Sigma
$$ We then have the following constructors for $\mathsf{CF}$ in cases where it is applied to a proof constructed from the identity or left/right rules: $$
\begin{array}{rcll}
\mathsf{cfId} & : & \forall a : \mathsf{Prop}\langle \rangle . \forall q_a : \mathsf{Ante}(a) . \\ & & \mathsf{CF} \langle a , \mathsf{id} \cdot a \triangleleft q_a \rangle & \in \Sigma\\
\mathsf{cf}{\multimap}\mathsf{R} & : & \forall a,b : \mathsf{Prop}\langle \rangle . \forall f : (\mathsf{Ante}(a) \multimap \mathsf{Conse}(b)) .\\ & & (\forall q : \mathsf{Ante}(a) . \mathsf{CF}\langle b, f \triangleleft q \rangle)\\ & & \multimap \mathsf{CF}\langle b, {\multimap}\mathsf{R} \cdot a \cdot b \triangleleft f \rangle & \in \Sigma\\
\mathsf{cf}{\multimap}\mathsf{L} & : & \forall a,b,c : \mathsf{Prop} \langle \rangle . \forall p_a : \mathsf{Conse}(a) .\\
& & \forall f : (\mathsf{Ante}(b) \multimap \mathsf{Conse}(c)) .\\ 
& & \forall q_{a \multimap b} : \mathsf{Ante}(a \multimap b) . \\
& & \mathsf{CF}\langle a, p_a \rangle\\ & & \multimap (\forall q_b : \mathsf{Ante}(b) . \mathsf{CF}\langle c , f \triangleleft q_b \rangle)\\ & & \multimap \mathsf{CF}\langle c , {\multimap}\mathsf{L} \cdot a \cdot b \cdot c \triangleleft p_a \triangleleft f \triangleleft q_{a \multimap b} \rangle & \in \Sigma
\end{array}
$$ To handle the case where $\mathsf{CF}$ is applied to a derivation whose outermost constructor is $\mathsf{cut}$, we further introduce the following relation to capture \emph{single-step cut reduction}: $$
\mathsf{CutStep}\left(\begin{array}{l} \bigoplus a,b : \mathsf{Prop}\langle \rangle .\\ (\mathsf{Conse}(a) \otimes (\mathsf{Ante}(a) \multimap \mathsf{Conse}(b))) \times \mathsf{Conse}(b) \end{array} \right) \in \Sigma
$$ Following Reed \cite{reed}, we use the product type former $\times$ to allow forming a derivation of type $\mathsf{Conse}(b)$ in the \emph{same} context as the corresponding inputs to cut, which are themselves represented with the type $\mathsf{Conse}(a) \otimes (\mathsf{Ante}(a) \multimap \mathsf{Conse}(b))$, which must therefore split up the context accordingly. We then add the following: $$
\begin{array}{rcll} \mathsf{cfCut} & : & \forall a, b : \mathsf{Prop}\langle \rangle .\\ & & \forall p : \begin{array}[t]{@{}l} ((\mathsf{Conse}(a) \otimes (\mathsf{Ante}(a) \multimap \mathsf{Conse}(b)))\\ ~ \times ~ \mathsf{Conse}(b)) . \end{array} \\ & & \mathsf{CutStep}\langle a , b , p \rangle\\ & & \multimap \mathsf{Cf}\langle b, \pi_2(p) \rangle\\ & & \multimap \mathsf{CF}\langle b, \mathsf{cut} \cdot a \cdot b \triangleleft \pi_1(p) \rangle & \in \Sigma \end{array}
$$ We encode the cut reduction procedure via axioms of the form $$\mathsf{CutStep}\langle a,b, (\langle s[p_1, \dots, p_n], f[p_{n + 1}, \dots, p_{n + m}] \rangle, t[p_1, \dots, p_{n + m}]) \rangle$$ where $p_1, \dots, p_{n + m}$ are some universally-quantified parameters. Note that the construction of $\mathsf{CutStep}$ forces these parameters to be used linearly in each of $s,f,t$. This is the key to the correct behavior of this representation of Cut Admissibility.

For instance, we have the following axiom for handling Principal Cuts (i.e. where a left rule meets a right rule): $$
\begin{array}{rcl}
\mathsf{prin}{\multimap} & : & \forall a,b,c : \mathsf{Prop} \langle \rangle .\\
& & \forall f : \mathsf{Ante}(a) \multimap \mathsf{Conse}(b) . \\
& & \forall p_a : \mathsf{Conse}(a)\\
& & \forall g : \mathsf{Ante}(b) \multimap \mathsf{Conse}(c)\\
& & \multimap \mathsf{CutStep}\langle a \multimap b , c,\\ & & \qquad (\langle {\multimap}\mathsf{R} \cdot a \cdot b \triangleleft f , \lambda q . {\multimap} \mathsf{L} \cdot a \cdot b \triangleleft p_a \triangleleft g \triangleleft q \rangle ,\\ & & \qquad \qquad \mathsf{cut} \cdot b \cdot c \triangleleft \langle \mathsf{cut} \cdot a \cdot b \triangleleft \langle p_a, f \rangle, g \rangle)\rangle
\end{array}
$$ The other axioms for various cases arising from Cut follow similarly. Note that the linearity constraints enforced upon the parameters to $\mathsf{CutStep}$ ensure that every derivation occurring as a parameter in the inputs to a cut must be used in the same quantity in the output of the cut reduction. 

It follows that Reed's examples \cite{reed} of spurious Cut Elimination rules that can be written in other linear logical frameworks do not apply to the above. In particular, Reed considers a case where, instead of the usual Right rule for $\multimap$, we instead had $$
\begin{array}{rcl}
{\multimap}\mathsf{Rbad}_2 & : & \forall a,b : \mathsf{Prop}\langle \rangle .\\ & & (\mathsf{Ante}(a) \multimap \mathsf{Ante}(a) \multimap \mathsf{Conse}(b))\\ & & \multimap \mathsf{Conse}(a \multimap b)
\end{array}
$$ In which case the corresponding cut reduction axiom for a principal cut would have to look something like $$
\begin{array}{l}
\mathsf{CutStep}\langle a \multimap b, c ,\\ \qquad (\langle {\multimap}\mathsf{Rbad}_2 \cdot a \cdot b \triangleleft f , \lambda q . {\multimap}\mathsf{L} \cdot a \cdot b \triangleleft p_a \triangleleft g \triangleleft q \rangle ,\\ \qquad \qquad \mathsf{cut} \cdot b \cdot c \triangleleft \langle \mathsf{cut} \cdot a \cdot b \triangleleft \langle p_a, (\lambda q_a . f \triangleright q_a \triangleright q_a) \rangle, g \rangle)
\end{array}
$$ but this is ill-typed, because the variable $q_a$ gets used twice in a term-level position in $f \triangleleft q_a \triangleleft q_a$. Similarly, if we instead had  $$
{\multimap}\mathsf{Rbad}_0 : \forall a,b : \mathsf{Prop}\langle \rangle . \mathsf{Conse}(b) \multimap \mathsf{Conse}(a \multimap b)
$$ then the corresponding cut reduction axiom for a principal cut would have to look something like $$
\begin{array}{l}
\mathsf{CutStep}\langle a \multimap b, c ,\\ \qquad (\langle {\multimap}\mathsf{Rbad}_0 \cdot a \cdot b \triangleleft f , \lambda q . {\multimap}\mathsf{L} \cdot a \cdot b \triangleleft p_a \triangleleft g \triangleleft q \rangle ,\\ \qquad \qquad \mathsf{cut} \cdot b \cdot c \triangleleft \langle \mathsf{cut} \cdot a \cdot b \triangleleft \langle p_a, (\lambda q_a . f)\rangle , g \rangle )
\end{array}
$$ but this is again ill-typed because now the variable $q_a$ does \emph{not} get used at the term level in the expression $f$.

Hence Reed's problem of representing a cut admissibility relation so as to allow for \emph{only} linear programs to be represented by this relation is solved in LLFDC$^\sigma$. From here, one may apply the usual structural induction on complexity of propositions / proofs involved in cuts to show that cut reduction terminates, and hence for every derivation $p_a$ of type $\mathsf{Conse}(a)$ there is a corresponding proof of type $\mathsf{CF}\langle a , p_a \rangle$.

\section{Conclusion \& Outlook}

The foregoing, I hope, constitutes a first step toward the theory of LFDCs and their internal language, hence also toward a \emph{substructural dependent type theory} at the right level of generality. Taking stock, we have seen that many of the constructs of ordinary dependent type theory (dependent pair/function types, etc.) can be given suitably-substructural analogues in this setting, and these enjoy many of the same metatheoretic desiderata, including type soundness and decidability. Moreover, these constructs are better-behaved than those of prior substructural dependent type theories, in that they do not suffer the same issues with representing substructural constraints in the formation of parameters to type families.

Of course, much remains to be done in fleshing out this theory. On the syntactic side of things, we may hope to extend the catalogue of constructs available in substructural dependent type theory with other mainstays of type theory, e.g. universes, inductive types, coinductive types, etc. On the semantic side, a fully-rigorous treatment of the informal semantics sketched in this paper is in order, including full definitions of LFDCs and associated constructs, along with a proof that the syntactic models of this type theory are themselves strict LFDCs that are equivalent to those in which they are interpreted. Moreover, the type theory of LFDCs should be applicable in \emph{all} LFDCs, not just strict ones, provided the following conjecture:

\begin{conjecture} Every LFDC (with type-level weakening, function types, products, etc.) is equivalent to a strict one.
\end{conjecture}

\subsection{Toward the type theory of monoidal topoi}

Beyond general development of the type theory LFDCs, a specific application of this theory is toward constructing a type theory for working internally in \emph{monoidal topoi}, i.e. topoi equipped with an additional (bi)closed monoidal structure. Such topoi arise naturally in the analysis of substructurally-typed programming languages -- if the types of a language $L$ form a monoidal category $\mathcal{M}_L$, then the presheaf category on $\mathcal{M}_L$ is a monoidal topos via Day Convolution, whose internal language is essentially the \emph{logic} of $L$-programs.

A type theory for such monoidal topoi must therefore combine aspects of ordinary dependent type theory, arising from the topos in the usual way, with the substructural aspect present in the topos due to its monoidal structure. Viewing this situation through the lens of LFDCs reveals that such monoidal topoi in fact consist of not one but \emph{two} LFDC structures, one Cartesian and given by the usual comprehension category structure on the topos, the other given by the monoidal structure on the topos as in \textbf{Ex. \ref{moncat-lfdc}}. What these two LFDC structures have in common is their shared category of contexts/closed types, i.e. the underlying topos. Generalizing this situation slightly, we may consider pairs of LFDCs whose categories of contexts/closed types are \emph{equivalent}. The functors constituting such an equivalence give ways of going back and forth between the type theories of these LFDCs in restricted contexts, and in this sense function as \emph{modalities} on these type theories.

This in turn suggests that the right way to type-theoretically capture such an equivalence is to make use of constructs from \emph{modal type theory} (c.f. \cite{multimodal}) in the type theory of LFDCs. Adapting such constructs from the usual categorical semantics of dependent type theory to the setting of LFDCs, and correspondingly from intuitionsitic to substructural dependent type theory, remains to be done. But the above reasoning suggests that if it can be carried out, such a marriage of modal and substructural type theory may yield significant benefits for the analysis of substructural programs.

Going even further, we may consider type theories interpreted not only in monoidal topoi, but monoidal $\infty$-topoi, i.e. models of Homotopy Type Theory (HoTT) equipped with a suitable notion of substructurality. In some ways, this is more natural from a computational point of view, since the internal language of 1-topoi is extensional Martin-Löf Type Theory, for which type checking is undecidable, while it is possible to give a type theory for HoTT possessing the normalization and canonicity properties \cite{sterling-angiuli}.

A closely related line of recent work is the form of linear dependent type theory devised by Mitchell Riley in his thesis \cite{riley}. This approach to linear dependent type theory makes use of ideas of bunched logic, and is based on a specific model of Homotopy Type Theory in parameterized spectra. Potential applications of this type theory in quantum certification have been further considered by e.g. Myers, Riley, Sati \& Schreiber \cite{quantum}. It remains to be seen whether and how this type theory is related to the form of substructural dependent type theory developed in this paper, i.e. whether one subsumes the other, etc. Such intertheoretic connections thus offer yet another avenue to be explored in developing the general theory of substructural dependent types.

\raggedright
\printbibliography

@InProceedings{atkey,
  title = {The Syntax and Semantics of Quantitative Type Theory},
  author = {Robert Atkey},
  year = {2018},
  doi = {10.1145/3209108.3209189},
  booktitle = {LICS '18: 33rd Annual ACM/IEEE Symposium on Logic in Computer Science, July 9--12, 2018, Oxford, United Kingdom}
}

@Inbook{mcbride,
author="McBride, Conor",
editor="Lindley, Sam
and McBride, Conor
and Trinder, Phil
and Sannella, Don",
title="I Got Plenty o' Nuttin'",
bookTitle="A List of Successes That Can Change the World: Essays Dedicated to Philip Wadler on the Occasion of His 60th Birthday",
year="2016",
publisher="Springer International Publishing",
address="Cham",
pages="207--233",
isbn="978-3-319-30936-1",
doi="10.1007/978-3-319-30936-1_12"
}

@incollection{martin-lof,
title = {An Intuitionistic Theory of Types: Predicative Part},
editor = {H.E. Rose and J.C. Shepherdson},
series = {Studies in Logic and the Foundations of Mathematics},
publisher = {Elsevier},
volume = {80},
pages = {73-118},
year = {1975},
booktitle = {Logic Colloquium '73},
issn = {0049-237X},
doi = {https://doi.org/10.1016/S0049-237X(08)71945-1},
author = {Per Martin-Löf}
}

@article{girard,
title = {Linear logic},
journal = {Theoretical Computer Science},
volume = {50},
number = {1},
pages = {1-101},
year = {1987},
issn = {0304-3975},
doi = {https://doi.org/10.1016/0304-3975(87)90045-4},
author = {Jean-Yves Girard}
}

@article{cervesato-pfenning,
title = {A Linear Logical Framework},
journal = {Information and Computation},
volume = {179},
number = {1},
pages = {19-75},
year = {2002},
issn = {0890-5401},
doi = {https://doi.org/10.1006/inco.2001.2951},
author = {Iliano Cervesato and Frank Pfenning}
}

@phdthesis{riley,
title = {A Bunched Homotopy Type Theory for Synthetic Stable Homotopy Theory},
author = {Mitchell Riley},
year = {2022},
doi = {https://doi.org/10.14418/wes01.3.139}
}

@phdthesis{reed,
title = {A Hybrid Logical Framework},
author = {Jason Reed},
year = {2009}
}

@inproceedings{krishnaswami-pradic-benton, 
author = {Krishnaswami, Neelakantan R. and Pradic, Pierre and Benton, Nick}, 
title = {Integrating Linear and Dependent Types}, 
year = {2015}, 
isbn = {9781450333009}, 
publisher = {Association for Computing Machinery}, address = {New York, NY, USA},
doi = {10.1145/2676726.2676969},
booktitle = {Proceedings of the 42nd Annual ACM SIGPLAN-SIGACT Symposium on Principles of Programming Languages}, 
pages = {17–30},
location = {Mumbai, India}, 
series = {POPL '15} }

@article{jacobs,
title = {Comprehension categories and the semantics of type dependency},
journal = {Theoretical Computer Science},
volume = {107},
number = {2},
pages = {169-207},
year = {1993},
issn = {0304-3975},
doi = {https://doi.org/10.1016/0304-3975(93)90169-T},
author = {Bart Jacobs},
}

@unpublished{cappucci-myers,
title = {The Para Construction as a Distributive Law},
author = {David Jaz Myers and Matteo Cappucci},
year = {2022},
note = {Talk given at Virtual Double Categories Workshop},
url = {https://bryceclarke.github.io/virtual-double-categories-workshop/slides/david-jaz-myers.pdf}
}

@inproceedings{multimodal,
author = {Gratzer, Daniel and Kavvos, G. A. and Nuyts, Andreas and Birkedal, Lars},
title = {Multimodal Dependent Type Theory},
year = {2020},
isbn = {9781450371049},
publisher = {Association for Computing Machinery},
address = {New York, NY, USA},
url = {https://doi.org/10.1145/3373718.3394736},
doi = {10.1145/3373718.3394736},
pages = {492–506},
location = {Saarbr\"{u}cken, Germany},
series = {LICS '20}
}

@INPROCEEDINGS{sterling-angiuli,
  author={Sterling, Jonathan and Angiuli, Carlo},
  booktitle={2021 36th Annual ACM/IEEE Symposium on Logic in Computer Science (LICS)}, 
  title={Normalization for Cubical Type Theory}, 
  year={2021},
  pages={1-15},
  doi={10.1109/LICS52264.2021.9470719}}

@unpublished{quantum,
title = {Effective Quantum Certification via Linear Homotopy Types},
author = {David Jaz Myers and Mitchell Riley and Hisham Sati and Urs Schreiber},
year = {2023},
url = {https://ncatlab.org/schreiber/files/QPinLHOTT-ExtendedAbstract-230315.pdf}
}

@article{debruijn,
title = {Telescopic mappings in typed lambda calculus},
journal = {Information and Computation},
volume = {91},
number = {2},
pages = {189-204},
year = {1991},
issn = {0890-5401},
doi = {https://doi.org/10.1016/0890-5401(91)90066-B},
url = {https://www.sciencedirect.com/science/article/pii/089054019190066B},
author = {N.G. {de Bruijn}}
}

\pagebreak

\appendix

\section{Auxiliary Definitions}

\begin{dfn}
The functor $\oplus_{\llp{\Delta}}$ associated to a telescope $\llp{\Gamma}, \llp{\Delta}$ is defined by recursion on $\llp{\Delta}$ as follows:
$$
\oplus_{\epsilon} = \text{Id}_{\mathfrak{T}(\llb{\llp{\Gamma}})} \qquad \oplus_{\llp{\Delta}, x : \llb{S}} = \oplus_{\llp{\Delta}} \circ \oplus_{\llb{S}}
$$ 
\end{dfn}

\begin{dfn}
Similarly, the substitution $f^*(\llp{\Theta})$ of a morphism $f : \llb{\llp{\Gamma}} \to \llb{\llp{\Delta}} \in \mathfrak{C}$ into a telescope $\llp{\Delta}, \llp{\Theta}$ is defined by recursion on $\llp{\Theta}$: $$
f^*(\epsilon) = \epsilon \qquad f^*(x : \llb{S}, \llp{\Theta}) = x : f^*(\llb{S}), f_\bullet(\llb{S})^*(\llp{\Theta})
$$ and likewise, the weakening $\omega_{\llp{\Gamma}}^{\llb{S}}(\llp{\Delta})$ of a telescope $\llp{\Gamma}, \llp{\Delta}$ is defined by recursion on $\llp{\Delta}$: $$
\omega_{\llp{\Gamma}}^{\llb{S}}(\epsilon) = \epsilon \quad \omega_{\llp{\Gamma}}^{\llb{S}}(y : \llb{T}, \llp{\Delta}) = y : \omega_{\llb{\llp{\Gamma}}}^{\llb{S}}(\llb{T}), \Omega_{\llp{\Gamma}}^{\llb{S}, \llb{T}}(\llp{\Delta})
$$ where $$
\Omega_{\llp{\Gamma}}^{\llb{S}, \llb{T}}(\epsilon) = \epsilon
$$ and $$
\begin{array}{l} \Omega_{\llp{\Gamma}}^{\llb{S}, \llb{T}}(z : \llb{R}, \llp{\Delta}) =\\ \quad z : \Omega_{\llb{\llp{\Gamma}}}^{\llb{S}, \llb{T}}(\llb{R}), (\alpha^{-1})^*(\Omega_{\llb{\llp{\Gamma}}}^{\llb{S}, \varoplus_{\llb{T}}(\llb{R})}(\alpha^*(\llp{\Delta})))
\end{array}
$$ 
\end{dfn}

\begin{dfn}
The type-level reassociation $${\bullet}\mathsf{asc}_{\llp{\Gamma}}^{\llp{\Delta}} : \llb{\llp{\Gamma}, \llp{\Delta}} \simeq \llb{\llp{\Gamma}}_\bullet(\llb{\llp{\Delta}})$$ of a telescope $\llp{\Gamma}, \llp{\Delta}$ is defined by recursion on $\llp{\Delta}$: $$
{\bullet}\mathsf{asc}_{\llparenthesis \Gamma \rrparenthesis}^\epsilon = \eta^{-1} \qquad {\bullet}\mathsf{asc}_{\llparenthesis \Gamma \rrparenthesis}^{x : \llbracket S \rrbracket, \llparenthesis \Delta \rrparenthesis} = {\bullet}\mathsf{asc}_{\llparenthesis \Gamma \rrparenthesis, x : \llbracket S \rrbracket}^{\llparenthesis \Delta \rrparenthesis} \circ \alpha^{-1}
$$ and likewise the term-level reassociation $${\oplus}\mathsf{asc}_{\llp{\Gamma}}^{\llp{\Delta} \mid \llb{S}} : \oplus_{\llp{\Delta}} (\llb{S}) \simeq \oplus_{\llb{\llp{\Delta}}}((({\bullet}\mathsf{asc}_{\llp{\Gamma}}^{\llp{\Delta}})^{-1})^*(\llb{S}))$$ $$
{\oplus}\mathsf{asc}_{\llp{\Gamma}}^{\epsilon \mid \llb{S}} = \ell^{-1} \qquad {\oplus} \mathsf{asc}_{\llp{\Gamma}}^{x : \llb{T}, \llp{\Delta} \mid \llb{S}} = \beta \circ \oplus_{\llb{T}}({\oplus}\mathsf{asc}_{\llp{\Gamma}, x : \llb{T}}^{\llp{\Delta} \mid \llb{S}})
$$
\end{dfn}

\begin{dfn}
The substitution ${\downarrow}(\llb{s}) : \llb{\llp{\Gamma}, \llp{\Delta}} \to \llb{\llp{\Gamma}}_\bullet(\llb{S})$ of a term $\llb{s} : \llb{\llp{\Delta}} \to \llb{S} \in \mathfrak{T}(\llb{\llp{\Gamma}})$ is defined as $$
{\downarrow}(\llb{s}) = \llbracket \llparenthesis \Gamma \rrparenthesis \rrbracket_\bullet(\llbracket s \rrbracket)\circ {\bullet}\mathsf{asc}_{\llparenthesis \Gamma \rrparenthesis}^{\llparenthesis \Delta \rrparenthesis}
$$ and ${\Downarrow}(\llb{s}) : \llb{\llp{\Gamma}, \llp{\Delta}, {\downarrow}(\llb{s})^*(\llp{\Theta})} \to \llb{\llp{\Gamma}, x : \llb{S}, \llp{\Theta}}$ is defined as $$
{\Downarrow}(\llb{s}) = {\bullet}\mathsf{asc}_{\llp{\Gamma}, x : \llb{S}}^{\llp{\Theta}} \circ {\downarrow}(\llb{s})_\bullet(\llb{\llp{\Theta}}) \circ ({\bullet}\mathsf{asc}_{\llp{\Gamma}, x : \llb{S}}^{\llp{\Theta}})^{-1}
$$
\end{dfn}

\begin{dfn}
The lifting functors $\uparrow_{\llp{\Gamma}}^{\llp{\Delta}} : \mathfrak{T}(\llb{\llp{\Gamma}}) \to \mathfrak{T}(\llb{\llp{\Gamma}, \llp{\Delta}})$ and $\Uparrow_{\llp{\Gamma}}^{\llp{\Delta} \mid \llb{S}} : \mathfrak{T}(\llb{\llp{\Gamma}}_\bullet(\llb{S})) \to \mathfrak{T}(\llb{\llp{\Gamma}, \llp{\Delta}}_\bullet(\uparrow_{\llp{\Gamma}}^{\llp{\Delta}}(\llb{S})))$ for a telescope $\llp{\Gamma}, \llp{\Delta}$ and type $\llb{S} \in \mathfrak{T}(\llb{\llp{\Gamma}})$ are defined by recursion on $\llp{\Delta}$ as follows: $$
\uparrow_{\llp{\Gamma}}^{\epsilon} = \text{Id}_{\mathfrak{T}(\llb{\llp{\Gamma}})} \qquad \uparrow_{\llp{\Gamma}}^{y : \llb{T}, \llp{\Delta}} = \uparrow_{\llp{\Gamma}, y : \llb{T}}^{\llp{\Delta}} \circ \omega_{\llb{\llp{\Gamma}}}^{\llb{T}}
$$ $$
\Uparrow_{\llparenthesis \Gamma \rrparenthesis}^{\epsilon, \llbracket S \rrbracket} = \text{Id}_{\mathfrak{T}(\llbracket \llparenthesis \Gamma \rrparenthesis \rrbracket_\bullet(\llbracket S \rrbracket))} \quad \Uparrow_{\llparenthesis \Gamma \rrparenthesis}^{y : \llbracket T \rrbracket, \llparenthesis \Delta \rrparenthesis \mid \llbracket S \rrbracket} = {\Uparrow}_{\llparenthesis \Gamma \rrparenthesis, y : \llbracket T \rrbracket}^{\llparenthesis \Delta \rrparenthesis \mid \omega_{\llbracket \llparenthesis \Gamma \rrparenthesis \rrbracket}^{\llbracket T \rrbracket}(\llbracket S \rrbracket)} \circ \Omega_{\llbracket \llparenthesis \Gamma \rrparenthesis \rrbracket}^{\llbracket T \rrbracket, \llbracket S \rrbracket}
$$ and the telescopic weakening functor $$\mathsf{W}_{\llp{\Gamma}}^{\llb{S} \mid \llp{\Delta}} : \mathfrak{T}(\llb{\llp{\Gamma}, \llp{\Delta}}) \to \mathfrak{T}(\llb{\llp{\Gamma}, x : \llb{S}, \omega_{\llp{\Gamma}}^{\llb{S}}(\llp{\Delta})})$$ is defined as$$
\mathsf{W}_{\llp{\Gamma}}^{\llb{S} \mid \llp{\Delta}} = ({\bullet}\mathsf{asc}_{\llp{\Gamma}, x : \llb{S}}^{\llp{\Delta}})^* \circ \Omega_{\llb{\llp{\Gamma}}}^{\llb{S}, \llb{\llp{\Delta}}} \circ (({\bullet}\mathsf{asc}_{\llp{\Gamma}, x : \llb{S}}^{\llp{\Delta}})^{-1})^*
$$
\end{dfn}

\begin{dfn}
The projection map $$\mathsf{proj}_{\llp{\Gamma}}^{\llb{\Delta} \mid \llb{S}} : \llb{\llp{\Delta}, x : {\uparrow}_{\llp{\Gamma}}^{\llp{\Delta}}(\llb{S}) }\xrightarrow[]{} \llb{S} \in \mathfrak{T}(\llb{\llp{\Gamma}})$$ for each telescope $\llp{\Gamma}, \llp{\Delta}$ and type $\llb{S} \in \mathfrak{T}(\llb{\llp{\Gamma}})$ is defined by recursion on $\llp{\Delta}$ as follows: $$
\mathsf{proj}_{\llparenthesis \Gamma \rrparenthesis}^{\epsilon \mid \llbracket S \rrbracket} = \text{id}_{\llbracket S \rrbracket} $$ $$ \mathsf{proj}_{\llparenthesis \Gamma \rrparenthesis}^{\llparenthesis \Delta \rrparenthesis, x : \llbracket T \rrbracket \mid \llbracket S \rrbracket} = \mathsf{proj}_{\llparenthesis \Gamma \rrparenthesis}^{\llparenthesis \Delta \rrparenthesis \mid \llbracket S \rrbracket} \circ \varoplus_{\llparenthesis \Delta \rrparenthesis}(\ell \circ \langle \top_{\llbracket T \rrbracket}, - \rangle)
$$
\end{dfn}
\end{document}